\documentclass{article}
\usepackage{ijcai23}
\usepackage{times}
\usepackage{soul}
\usepackage{amsfonts}
\usepackage{url}
\usepackage[hidelinks]{hyperref}
\usepackage[utf8]{inputenc}
\usepackage[small]{caption}
\usepackage{graphicx}
\usepackage{amsmath}
\usepackage{amsthm}
\usepackage{booktabs}
\usepackage{algorithm}
\usepackage{algorithmic}
\usepackage{color}
 \usepackage{booktabs}
\usepackage{float}
\usepackage{graphicx}
\usepackage{caption}
\usepackage{tikz}
\usepackage{pgfplots}
\usepackage{placeins}

\graphicspath{ {./images/} }

\newtheorem{obs}{Observation}

\usepackage[lofdepth,lotdepth]{subfig}
\usepackage{etoolbox}
\usepackage{tabularx}

\newtheorem{theorem}{Theorem}
\newtheorem{lemma}{Lemma}

\newcommand{\charlotte}[1]{\textcolor{red}{{C:}#1}}

\newcommand{\sajjad}[1]{\textcolor{green}{{(Sajjad:} #1)}}

\title{Why Rumors Spread Fast in Social Networks, and How to Stop It\footnote{This paper appears in 32nd International Joint Conference on Artificial Intelligence (IJCAI-2023).}}

\author{
Ahad N. Zehmakan$^1$
\and
Charlotte Out$^2$\and
Sajjad Hesamipour Khelejan$^{3}$
\affiliations
$^1$ School of Computing, The Australian National University\\
$^2$ Department of Computer Science \& Technology, University of Cambridge\\
$^3$ School of Computer Science \& Statistics, Trinity College Dublin
\emails
ahadn.zehmakan@anu.edu.au,
ceo33@cam.ac.uk,
hesamips@tcd.ie
}
\begin{document}
\maketitle
\begin{abstract}
We study a rumor spreading model where individuals are connected via a network structure. Initially, only a small subset of the individuals are spreading a rumor. Each individual who is connected to a spreader, starts spreading the rumor with some probability as a function of their trust in the spreader, quantified by the Jaccard similarity index.
Furthermore, the probability that a spreader diffuses the rumor decreases over time until they fully lose their interest and stop spreading.

We focus on determining the graph parameters which govern the magnitude and pace that the rumor spreads in this model. We prove that for the rumor to spread to a sizable fraction of the individuals, the network needs to enjoy ``strong'' expansion properties and most nodes should be in ``well-connected'' communities. Both of these characteristics are, arguably, present in real-world social networks up to a certain degree, shedding light on the driving force behind the extremely fast spread of rumors in social networks.

Furthermore, we formulate a large range of countermeasures to cease the spread of a rumor. We introduce four fundamental criteria which a countermeasure ideally should possess. We evaluate all the proposed countermeasures by conducting experiments on real-world social networks such as Facebook and Twitter. We conclude that our novel decentralized countermeasures (which are executed by the individuals) generally outperform the previously studied centralized  ones (which need to be imposed by a third entity such as the government).
\end{abstract}
\vspace{-0.2cm}
\section{Introduction}

With the rapid development of the Internet, social media has become a convenient online platform for users to obtain information, express and exchange opinions and stay in touch with friends. However, online social networks also pave the road for the propagation of misinformation, particularly rumors (commonly defined as unverified information or deliberately falsified news). 
It is usually difficult for the public to recognize the falsehood of a rumor, especially if it is designed skillfully, cf.~\cite{vosoughi2018spread}. 
The spread of rumors can mislead people to behave in irrational ways, which can cause a series of undesirable consequences, such as public panic, virtual assets losses, manipulation of the outcome of political events, and economic damages. 
Consequently, there has been a growing demand and interest to gain insights into the rumor spreading dynamics and design powerful countermeasures to reduce the threats posed by rumors.

To shed some light on the fundamental characteristics and essential principles of rumor propagation phenomenon, scholars from a vast spectrum of backgrounds have introduced and studied various rumor spreading models, such as SIR model~\cite{zhao2012sihr}, Push-Pull protocol~\cite{giakkoupis2011tight}, DK model~\cite{daley1965stochastic} and the Independent Cascade (IC) model ~\cite{kempe2003maximizing}. In most of these models, the interactions and influence between the individuals are modelled using a graph structure, which represents a social network (SN). The state of the individuals (e.g., informed/uninformed) is updated following a rumor spreading rule. The updating rules are tailored to capture various properties observed in real-world scenarios, usually conceptualized by social scientists. 

One aspect of the rumor spreading dynamics which has attracted a substantial amount of attention is the design of effective countermeasures to stop or slow down the spread of rumors, e.g., blocking users, blocking connections, and spreading an ``anti-rumor'', cf.~\cite{he2015modeling}.

In the present paper, we introduce a rumor spreading model which inherits characteristics of the IC, Push-Pull, and SIR model and additionally captures the well-established sociological concepts of the impact of trust and forgetting mechanism. In an attempt at fighting rumor spreading, we study six countermeasures. We establish four essential criteria that a good countermeasure should possess
and evaluate the proposed countermeasures on those. We observe that the decentralized countermeasures perform generally better.


\subsection{Preliminaries}
\label{Preliminaries}

Let $G=\left(V,E\right)$ be a simple connected undirected graph, where $n:=|V|$ and $m:=|E|$.  
For a node $v\in V$, $N\left(v\right):=\{v'\in V: \{v',v\} \in E\}$ is the \emph{neighborhood} of $v$. Furthermore, $\hat{N}(v):=N(v)\cup \{v\}$ is the \textit{closed neighborhood} of $v$. Let $d\left(v\right):=|N\left(v\right)|$ be the \emph{degree} of $v$ in $G$. We also define $d_A(v):=|N(v)\cap A|$ for a set $A\subseteq V$. Furthermore, let $\partial(A):= \{v\in V\setminus A : \{v',v\}\in E, v' \in A\}$ be the \textit{node boundary} of $A\subset V$.

We define a \emph{coloring} to be a function $\mathcal{C}:V\rightarrow\{r,u,o\}$, where $r$, $u$, and $o$ represent \textit{red}, \textit{uncolored}, and \textit{orange} respectively. For a node $v\in V$, the set $N_a^{\mathcal{C}}\left(v\right):=\{v'\in N\left(v\right):\mathcal{C}\left(v'\right)=a\}$ includes the neighbors of $v$ which have color $a\in\{r,u, o\}$ in coloring $\mathcal{C}$.

\paragraph{Rumor Spreading Model.}
Consider an initial coloring $\mathcal{C}_0$ of a graph $G$. In each round, all nodes simultaneously update their color according to the following updating rule:
\vspace{0.2cm}

$\mathcal{C}_t(v)$ =
$\begin{cases}
o & \text{if } \mathcal{C}_{t-1}(v)=o \\
r & \text{if } \mathcal{C}_{t-1}(v)=r \text{ and } \mathcal{J}_t(v) < k \\
o & \text{if } \mathcal{C}_{t-1}(v)=r \text{ and } \mathcal{J}_t(v) = k \\
u & \text{if } \mathcal{C}_{t-1}(v)=u \text{ w.p. } p^*(v)\\
r & \text{if } \mathcal{C}_{t-1}(v)=u \text{ w.p. }1-p^*(v)
\end{cases}$

where $\mathcal{C}_t\left(v\right)$ is the color of node $v$ in the $t$-th round, integer $k$ is a model parameter, $\mathcal{J}_t(v)$ is the number of rounds $v$ has been red until round $t$, $\mathcal{S}(v,v^{\prime}):=|\hat{N}(v)\cap \hat{N}(v')|/|N(v)\cup N(v^{\prime})|$ for $v,v^{\prime}\in V$, and $p^*(v):= \prod_{v^{\prime}\in N_{r}^{\mathcal{C}_{t-1}}(v)} \left(1-\frac{S(v,v^{\prime})}{2^{\mathcal{J}_{t}(v')}}\right)$. A red node corresponds to an individual who is \emph{informed} of the rumor. An informed node stops spreading the rumor after $k$ rounds and turns orange (\emph{uninterested}), which it remains forever. An uncolored node corresponds to an \emph{uninformed} individual. If an uninformed (uncolored) node $v$ is adjacent to an informed (red) node $v'$, then $v'$ turns $v$ into red w.p. $S(v,v^{\prime})/2^{\mathcal{J}_{t}(v')}$ independently. Thus, $v$ becomes red in the next round w.p. $1-p^*$ and remains uncolored w.p. $p^*$. The coefficient $1/2^{\mathcal{J}(v^{\prime})}$ corresponds to the probability that $v'$ spreads the rumor and $\mathcal{S}(v,v')$ is the probability that $v$ accepts it. The value of $1/2^{\mathcal{J}(v^{\prime})}$ accounts for the fact that $v'$ might not necessarily spread the rumor w.p. 1 and the probability decreases exponentially in the number of rounds $v'$ has been informed of the rumor, reflecting the fact that an individual loses interest in a rumor over time, cf.~\cite{zhao2013rumor}. The coefficient $\mathcal{S}(v,v')$ (the Jaccard index) which measures the similarity between two nodes reflects the fact that people are more likely to accept information from their trusted connections ~\cite{figeac2021behavioral}.
In the numerator, we use $\hat{N}(v)$. This is to ensure that for two adjacent nodes $v,u$, $S(v,u)$ (the accepting probability) is not zero. We could analogously define $S(v,u)=(|N(v)\cap N(u)|+2)/|N(v)\cup N(u)|$ since we are always concerned about adjacent nodes. We note that we can also view $S(v,v^{\prime})$ as the weight of the edge $\{v,v'\}$.

Our model is different from the IC model in two ways: In the IC model (i) $k$ is always set to 1 (i.e., a red node becomes orange after one round) (ii) the weights are usually assigned randomly.

Starting from any initial coloring, the process eventually reaches a \emph{fixed} coloring where all nodes are orange or uncolored. If the process reaches a coloring with a constant fraction of orange nodes, say $10\%$, then we say that the rumor \textit{spreads}, and it does not otherwise. (There is nothing unique about $10\%$ and our results hold for similar fixed values.)

\paragraph{Graphs.} 
Let $\mathcal{G}_{n,p}$ denote the Erd\H{o}s-R\'{e}nyi (ER) random graph, which is the random graph on $n$ nodes, where each edge is present independently w.p. $p$. For integers $n$ and $r$, we define the ($n$,$r$)-\textit{flower} graph in the following way. Consider a cycle $C_N=v_1,\cdots, v_N$ for $N=n/r$. For each node $v_i$, add a distinct clique of size $r-1$ to the graph and add an edge between $v_i$ and every node in the clique. We refer to each node $v_i$ and its clique as a \textit{super node} and $v_i$ is called the \textit{boundary} node of the super node. We are particularly interested in the case of $r=\log^2(n)$, which is simply called the $n$-\textit{flower} graph. (Note that ($n$, $1$)-flower graph is simply a cycle graph with $n$ nodes.)

To measure the expansion of a graph, we consider an algebraic characterization of expansion. Let $\lambda(G)$ be the second-largest absolute eigenvalue of the adjacency matrix of $G$. Small values of $\lambda(G)$ imply that $G$ has strong expansion properties (i.e., is well-connected). For integers $n,d$, we define the ($n$, $d$)-\emph{moderate expander} graph in the following way, where we always assume that $n$ is ``significantly'' larger than $d$. Let $H$ be a $N$-node, $D$-regular graph such that $\lambda(H) \leq C\sqrt{D}$, where $N = \frac{n}{\log^2(n)}$ and $D = d\cdot \log^2(n)$ and $C$ is a positive constant. Replace every node $x$ in $H$ with a clique of size $\log^2 n$ and then evenly distribute the $D$ edges of $x$ among these $\log^2 n$ nodes. The obtained $n$-node ($\log^2 n+d-1$)-regular graph is an ($n$,$d$)-moderate expander, which is denoted by $\mathcal{M}_{n,d}$. Similar to the $(n,r)$-\emph{flower} graph, the set of $\log^2 n$ nodes in each of the $N$ cliques is called a \textit{super node}.
(Note that moderate expanders are not meant to mimic real-world SNs. They are solely designed to maximize the spread of rumors and are objects of theoretical interest.) 
\paragraph{Experimental Setup.}
\label{experimental setup}
For our experiments, we rely on publicly available graph data from~\cite{snapnets}. Our experiments were conducted on the following SNs: Twitter ($81306$ nodes and $1342310$ edges), Facebook ($4039$ nodes and $88234$ edges), Google+ ($107614$ nodes and $13673453$ edges), Twitch Germany ($9498$ nodes and $153138$ edges), and Twitch France ($6549$ nodes and $122666$ edges). We use shorthand TW, FB, G+, T-GE, T-FR, respectively. We also conducted experiments on ER random graph and Hyperbolic random graph (HRG). 
The parameters in these graphs were set such that the (expected) number of nodes/edges is comparable to the ones in the aforementioned real-world networks. For HRG, one also needs to provide the exponent of the power-law degree distribution $\beta$ and the temperature $T$ as the input parameters. We set $\beta= 2.5$ and $T = 0.6$. We used the algorithm of~\cite{staudt2015networkit} for the generation of HRG random graphs. Furthermore, the experiments which required random choice of edges or colors were executed 100 times
and then the average output was considered. The standard deviations are reported in Appendix \ref{appendix-plots}. The code for the experiments is available at \url{https://github.com/charlotteout/RumourSpreading}.

\paragraph{Assumptions.} All logarithms are to base $e$, unless pointed out otherwise. We let $n$ tend to infinity and say an event $\mathcal{E}$ happens with high probability (w.h.p.) if it occurs w.p. $1-o(1)$. We always assume that initially one randomly chosen node is red, and all other nodes are uncolored, otherwise, it is stated explicitly.
Furthermore, we suppose the parameter $k$ is a small integer, say $k=5$, but our results would hold for any constant value of $k$. 

\subsection{Our Contribution}
We study a rumor spreading model which captures fundamental characteristics  such as the randomized spreading mechanism and various agent types as introduced in the IC, Push-Pull and SIR model, as well as sociological concepts such as the impact of homophily on trust ~\cite{granovetter1973strength}, formulated by the Jaccard index, and the forgetting mechanism~\cite{zhao2013rumor}.

Firstly, we address the question: What are the graph structures for which the rumor spreads (in other words, what graph parameters govern the spread of rumors)? It has previously been argued that that information disseminates quickly when the graph has strong expansion properties (i.e., is well-connected),
cf.~\cite{SauerwaldVertexexpansion}. 
However, for our model expansion is not solely sufficient for a rumor to spread, especially if the graph is sparse which is usually the case in the real-world SNs. In particular, we prove that in our model on the ER random graph $\mathcal{G}_{n,p}$ (which enjoys strong expansion properties, cf.~\cite{le2017concentration}) for $p$ sufficiently smaller than $1/\sqrt{n}$, the rumor does not spread with a constant probability.


Additionally, we show that an abundance of very well-connected local communities (which result in  large values of $\mathcal{S}(v,v')$ for adjacent nodes $v,v'$) alone also cannot guarantee extensive spread of rumors. In particular, we prove that on an ($n$, $r$)-flower graph, where $\mathcal{S}(v,v')=1$ for almost every two adjacent nodes $v,v'$, the rumor does not spread w.h.p. for $r\le n^{1-\epsilon}$ and $\epsilon>0$ (even when we start with $o(\log n)$ red nodes).

However, we show that the combination of these two properties guarantees an extremely fast spread of rumors. More precisely, we prove for even very sparse moderate expander graphs, the rumor spreads in logarithmically many rounds. Roughly speaking, the strong local communities help the rumor to spread quickly inside a community once it reaches a node in that community and expansion ensures that it breaks out into other communities invasively. (We emphasize that the average degree of moderate expanders in this set-up is in the order of $\log^2 n$, which is much smaller than the average degree of $\sqrt{n}$ required in ER graphs for spreading.)

A natural question to ask is whether the rumor spreads on real-world SNs in our model. Our experiments on real-world graph data such as Twitter and Facebook demonstrate that the rumor indeed spreads to a very large body of the network in a short period of time. (A visualization of the process on the Facebook SN is given in Appendix~\ref{appendix-visualize}.) While the social graphs which emerge in the real world do not have the expansion and community structure tailored for the moderate expanders, they still enjoy a certain level of expansion, and well-connected communities are present in abundance. Note that this is an indication that our model is more realistic than previous models such as Push-Pull models, which advocate strong expansion properties as necessary and sufficient condition for fast spread of rumors, as we know that in real life rumors spread very fast in real-world SNs, and they are not strong expanders.

Moreover, we formulate and investigate several countermeasures. Some of them (e.g., blocking nodes and edges) need to be implemented by a third entity such as the government, and we refer to them as centralized countermeasures. On the other hand, the decentralized ones are executed by the members of the network. 
It turns out that the proposed decentralized countermeasures not only enjoy several desirable criteria such as not interfering with freedom of expression and not being too intrusive, but also significantly outperform the centralized ones in stopping the spread of the rumor according to our experiments. The prior work has focused on the development of centralized countermeasures, see Section~\ref{prior-work} (which are also implemented in practice up to some degree, e.g., by blocking accounts). Our work aspires to send out the message that the focus should be shifted towards the development of decentralized countermeasures, which can be achieved for instance through educating the members rather than forceful actions of a third entity.

\subsection{Prior Work}\label{prior-work}

A plethora of rumor spreading models have been developed and studied in recent years, cf.~\cite{n2020rumor,zehmakan2023random}. Here, we focus on the most fundamental and relevant models, which have inspired our work.

\paragraph{Push-Pull Models.} In this set-up, each node is either red or uncolored. In each round, every red node makes a randomly chosen neighbor red (Push model), or every uncolored node adopts the color of a randomly chosen neighbor (Pull model), or both (Push-Pull model). Since there is no forgetting mechanism in place, all nodes eventually become red (i.e., the rumor spreads). Thus, a natural question is how long this takes. For the Push model, the spreading time is known~\cite{feige1990randomized} to be $\mathcal{O}(\Delta \cdot (\Lambda + \log(n))$, where $\Delta$ and $\Lambda$ are the maximum degree and diameter of the underlying graph. For the Push-Pull model, after a long line of research, the bound $\mathcal{O}\left(\Phi^{-1} \log(n)\right)$, for $\Phi$ being the conductance of the graph, was proven~\cite{SauerwaldVertexexpansion}.

\paragraph{Independent Cascade (IC) Model.}
In the IC model~\cite{goldenberg2001talk}, in each round every red node $v$ makes an uncolored node $u$ in its neighborhood red w.p. $p_{vu}$. A red node becomes orange after one round, which is similar to setting $k=1$ in our model. However, in the IC model, the probabilities $p_{vu}$ are chosen uniformly at random. Motivated by viral marketing, the main focus in this model is developing algorithms for finding subsets of nodes that maximize the spread of the red color, mostly exploiting monotonicity and submodularity properties (cf.~\cite{mossel2007submodularity,chen2011influence}).


\paragraph{Weighted Connections.} 
Recall that in the IC model (and other similar models) weights are assigned to the edges randomly. As this is not entirely realistic, it would be relevant to introduce meaningful weight assignment mechanisms. Using the communication information of individuals on various real-world networks,~\cite{onnela2007structure} and~\cite{goyal2010learning} observed that there is a strong correlation between the number of shared friends of two individuals and their level of communication. Consequently, they proposed the usage of similarity measures, such as Jaccard-like parameters, to approximate the weights of connections between nodes. This is also aligned with the well-studied strength of weak ties hypothesis~\cite{granovetter1973strength}.
This line of research has inspired the choice of Jaccard index in our model.

\paragraph{Countermeasures.} A large part of the research efforts for developing countermeasures is concentrated around 
 blocking nodes and edges. However, since in most models finding the most ``influential'' nodes/edges is NP-hard, cf.~\cite{kempe2003maximizing}, the focus has been on approximate blocking strategies, which use structural properties. For nodes, various algorithms such as blocking nodes with the highest degree, betweenness, and closeness have been investigated, cf.~\cite{he2015modeling,wang2015maximizing,yu2008finding}. Furthermore, for different greedy-based edge blocking strategies to minimize the spread in the IC model, see~\cite{kimura2008minimizing,yan2019rumor}. Other studied countermeasures are spreading the truth as an anti-rumor, cf.~\cite{tripathy2010study,ding2020efficient}, inoculation strategies (which rest on the idea that if people are forewarned that they might be misinformed, they become more immune), cf.~\cite{lewandowsky2021countering}, and accuracy flags, cf.~\cite{gausen2021can}. For more results on countermeasures also see~\cite{coroexploiting,bredereck2021maximizing,zhengmfan,zehmakan2019spread,qian2018neural,zehmakan2019tight,ma2016detecting,zehmakan2021majority}.

\section{When Does a Rumor Spread?}\label{why-section}
\subsection{Erd\H{o}s-R\'{e}nyi Random Graph}
\begin{theorem}\label{er-thm}
Consider the coloring where only one node is red (the rest is uncolored) on $\mathcal{G}_{n,p}$ with $p \leq \frac{1}{n^{\frac{1}{2} + \epsilon}}$ for any constant $\epsilon > 0$. The rumor does not spread with a constant probability.
\end{theorem}

\begin{proof}
Define $s:=\lceil1/\epsilon\rceil+1$.
For a pair of distinct nodes $v$ and $u$, the probability that the inequality $|N(v)\cap N(u)|\ge s$ holds is upper-bounded by $n-2 \choose s$ $p^{2s}$. Let $X$ be the number of pairs which satisfy the above inequality. Then, we have $\mathbb{E}[X] \le \binom{n}{2} \binom{n-2}{s}p^{2s}\leq n^{s+2}p^{2s} \leq \frac{n^{s+2}}{n^{s+2s\epsilon}} = o(1)$, where we used that $p\le 1/n^{\frac{1}{2}+\epsilon}$ and $s\epsilon >1$, respectively. Hence, by Markov's inequality (Lemma~\ref{Markov} in Appendix~\ref{appendix-ineq}), $\Pr[\mathcal{A}]=\Pr[X \geq 1] \le o(1)$, where $\mathcal{A}$ is the event that $X\ge 1$ (and $\bar{\mathcal{A}}$ is the complement of $\mathcal{A}$).

\noindent Let $v$ be the only node which is colored red in $\mathcal{C}_0$. For each node $u \in N(v)$, we have $\Pr[\mathcal{C}_1(u) = r| d(v) = d \land \bar{\mathcal{A}}] = \frac{|\hat{N}(v) \cap \hat{N}(u)|}{2|N(v) \cup N(u)|} \leq \min\left(\frac{s+2}{2d}, \frac{1}{2}\right)$. For $(s+2)/(2d)$, we used that $|\hat{N}(v)\cap \hat{N}(u)|\le|N(v)\cap N(u)|+2\le s+2$ and $|N(v)\cup N(u)|\ge d(v)=d$. The upper bound of $1/2$ holds because $|\hat{N}(v)\cap \hat{N}(u)|\le|N(v)\cup N(u)|$.

\noindent Let $\mathcal{E}_i$, for $1\le i\le k$, denote the event that $v$ does not make any of its neighbors red in the $i$-th round. Then, $ \Pr[\mathcal{E}_1 |d(v) = d \land \bar{\mathcal{A}}] \geq \left(1 - \min\left(\frac{s+2}{2d},\frac{1}{2}\right) \right)^d$. If $\frac{s+2}{2d} < 1/2$ then $\left(1 - \frac{s+2}{2d} \right)^d \geq \left(\frac{1}{4}\right)^{(s+2)/2}$ (which gives a constant lower bound) using the estimate $(1-x) \geq \left(\frac{1}{4}\right)^{x}$ for $x < 1/2$. If $\frac{s+2}{2d}\geq 1/2$, then $d \leq s+2$, which implies that $(1/2)^d$ is a constant. Therefore, in both cases, we can lower bound $\left(1 - \min\left(\frac{s+2}{2d},\frac{1}{2}\right)\right)^d$ with some constant $C>0$.

\begin{align*}
    \Pr[\mathcal{E}_1] & =  \Pr[\mathcal{\bar{\mathcal{A}}}] \cdot \Pr[\mathcal{E}_1|\bar{\mathcal{A}}] + \Pr[\mathcal{\mathcal{A}}] \cdot \Pr[\mathcal{E}_1|\mathcal{A}] \ge \\ & \Pr[\mathcal{\bar{\mathcal{A}}}]\cdot \sum_{d=0}^{n-1}\Pr[\mathcal{E}_1|d(v) = d \land \bar{\mathcal{A}}]\cdot \Pr[d(v) = d]
    \ge \\&\Pr[\mathcal{\bar{\mathcal{A}}}]\cdot \sum_{d=0}^{n-1}C\cdot \Pr[d(v) = d] = (1-o(1))\cdot C  \geq \frac{C}{2}.
\end{align*}

\noindent With a similar argument, we can prove that $\Pr[\mathcal{E}_i|\mathcal{E}_{i-1}\land \cdots \land \mathcal{E}_1]\ge C/2$ for $2\le i\le k$. Thus, we have $\Pr[\mathcal{E}_1 \land \cdots \land \mathcal{E}_k] =  \Pr[\mathcal{E}_k | \mathcal{E}_{k-1} \land \dots \land \mathcal{E}_1] \cdots \Pr[\mathcal{E}_2|\mathcal{E}_1]\cdot\Pr[\mathcal{E}_1] \geq \left(C/2\right)^k$.
This implies that w.p. at least $(C/2)^k=(C/2)^5$, no node becomes red during the first $k$ rounds. In that case, the process ends with one orange node and $n-1$ uncolored nodes in $k$ rounds. This bound on $p$ turns out to be tight, please refer to Appendix~\ref{appendix-ER-tightness} for a full proof.
\end{proof}


\subsection{Flower Graph}
A super node whose all nodes are uncolored is called \textit{uncolored} and \textit{colored} otherwise. And it is said to be \textit{red} if all its nodes are red.
\begin{theorem}
\label{flower-thm}
\label{(n,r)-cycle proof}
Consider an ($n$,$r$)-flower graph for $r\le n^{1-\epsilon}$ and constant $\epsilon>0$. If initially there are $s(n)=o(\log n)$ red super nodes (and the rest is uncolored), the rumor does not spread w.h.p.
\end{theorem}
\textit{Proof Sketch.}
 A path of super nodes is a sequence of super nodes which form a path in the cycle obtained from collapsing each super node into a node. A path is \textit{uncolored} if all its super nodes are uncolored. In a \textit{reddish path}, there are no two adjacent uncolored super nodes and the endpoints are colored. We note that for any coloring of the ($n$,$r$)-flower graph, there is a set of maximal uncolored and reddish paths which partition the nodes in the graph.

Define a \textit{phase} to be a sequence of $k$ rounds. Let $\mathcal{C}$ be the coloring at the beginning of phase $i$.
Consider all the endpoints of the uncolored paths in the aforementioned partitioning and define $U$ to be their boundary nodes. Let $\mathcal{E}_i$ be the event that no node in $U$ becomes red during the whole phase.

We observe that if the event $\mathcal{E}_i$ occurs, then all boundary nodes of the reddish paths endpoints become orange. Thus, all nodes which are not on any reddish path remain uncolored forever. Let us define $t^*:= (1/C)^{2s(n)} \log(n)$, for a suitably chosen constant $0<C<1$, then with some relatively straightforward calculations, we can show that $\Pr[\land^{t^*}_{i=1}\bar{\mathcal{E}}_{i}] \le \frac{1}{n}$. Thus, w.h.p. after at most $t^*$ phases (i.e., $kt^*$ rounds), we reach a coloring where all nodes which are not on any reddish path remain uncolored forever. Furthermore, we claim that the number of nodes on the reddish paths during the first $kt^*$ rounds is sub-linear. Hence, the rumor does not spread w.h.p. A full proof is given in Appendix~\ref{appendix-flower}. \qed

\subsection{Moderate Expander}\label{me-section}
\begin{theorem}\label{me-thm}
\label{rumorspreadsonmoderateexpander}
Consider an ($n$, $d$)-moderate expander $\mathcal{M}_{n,d}$ with $d=\omega(1)$. If initially there is a red node (and the rest are uncolored), then the rumor spreads w.h.p. in $\mathcal{O}(\log_d n)$ rounds.
\end{theorem}

Similar to a flower graph, we call a super node $x$ \textit{uncolored} if all its nodes are uncolored. We say $x$ is \textit{strong red} if every node in it has become red at most three rounds before. A super node is \textit{weak red} if it is neither strong red nor uncolored. Let $u_t$, $s_t$ and $w_t$ denote the number of uncolored, strong red, and weak red super nodes in the $t$-th round.

Recall that if we contract all $N=n/\log^2 (n)$ super nodes in $\mathcal{M}_{n,d}$, we obtain a $D$-regular graph for $D=d\log^2(n)$. In Lemma~\ref{spreadfastEstar} (proven in Appendix~\ref{appendix-lemma-clique}), we state that if a node in one of these super nodes is red, then the super node becomes red in 2 rounds. Then, in Lemma~\ref{me-lemma}, we show that the number of strong red super nodes increases by roughly a $d$ factor after every three rounds. Repeated application of Lemma~\ref{me-lemma} implies that the rumor spreads in $\mathcal{O}(\log_d n)$ rounds. (A more detailed discussion is given in Appendix~\ref{appendix-me-thm}, where we also argue that the bound $d=\omega(1)$ is necessary, i.e., the statement does not hold for constant $d$).

\begin{lemma}
\label{spreadfastEstar}
Consider a graph $G=(V,E)$ where nodes in $\mathcal{K}\subseteq V$ form a clique, $\kappa := |\mathcal{K}|\ge \log^2 n$, and for every $v\in \mathcal{K}$ $d(v)\le 2\kappa$. If $\mathcal{C}_t(v)=r$ for some $v\in \mathcal{K}$ and all other nodes in $\mathcal{K}$ are uncolored, then there is no uncolored node in $\mathcal{K}$ in round $t+2$ w.p. $1 - o(1/n)$.
\end{lemma}
To prove Lemma~\ref{me-lemma}, we need Lemma~\ref{constantboundary} and Observation~\ref{obvs1}. The proof of Lemma~\ref{constantboundary} is given in Appendix~\ref{appendix-lemma-constantboundary}, which relies on the expander mixing lemma, cf.~\cite{friedman2003proof_ALONEIGVALUE}.
\begin{lemma}
\label{constantboundary}
Consider an $N$-node $D$-regular graph $G$, where $\lambda \leq C\sqrt{D}$, for some constant $C>0$, and $D=\omega(1)$. If a node set $A$ is of size at most $\frac{N}{10}$, then there is some constant $C'>0$ such that $|\partial(A)| \geq \min\left(2N/5, |A|C'D \right)$.
\end{lemma}

\begin{obs}
\label{obvs1}
Let $x$ and $y$ be two distinct super nodes in a moderate expander graph. Then, there is at most one edge between $x$ and $y$, by construction.
\end{obs}

\begin{lemma}\label{me-lemma}
Consider an ($n$, $d$)-moderate expander $\mathcal{M}_{n,d}$ with $d=\omega(1)$. If $1\le s_t<C_1N/D$, for a sufficiently small constant $C_1>0$, and $w_t=\mathcal{O}(s_t/d)$, then after three rounds there are $\Omega(s_td)$ new strong red super nodes w.p. $1-\exp(\Omega(-ds_t))-o(1/\log n)$.
\end{lemma}

\begin{proof}
Let $\mathcal{E}^*$ be the event that every uncolored super node becomes strong red in two rounds once it has at least one red node. Based on Lemma~\ref{spreadfastEstar}, $\mathcal{E}^*$ holds w.p. at least $1-N\cdot o(1/n)\ge 1-o(1/\log n)$ since there are $N$ super nodes.

Furthermore, let $q^*$ denote the probability that a node $v$, in a strong red super node, makes a node $u$, in an uncolored super node, red where there is an edge between $v$ and $u$. Since $|\hat{N}(v)\cap \hat{N}(u)|\ge 2$, $|N(v)\cup N(u)|\le 2(d+\log^2 n)\le 2.5 \log^2 n$ (using the assumption that $d$ is significantly smaller than $n$), and $v$ has been red for at most three rounds, we get the following upper-bound:
\begin{equation}\label{eq-q*}
q^*\ge \frac{|\hat{N}(v)\cap \hat{N}(u)|}{2^3|N(v)\cup N(u)|}\ge \frac{2}{8\times 2.5 \log^2 n}= \frac{1}{10\log^2 n}.
\end{equation}

Let $S$, $W$, and $U$ be the set of strong red, weak red, and uncolored super nodes in round $t$. Let us label the nodes in $\partial(S)\cap U$ from $u_1$ to $u_b$, where $b$ is the size of $\partial(S)\cap U$. For each node $u_i$ consider one of its neighbors in $S$. Let Bernoulli random variable $y_i$ be 1 if and only if $u_i$ is made red by that neighbor in $S$ in the next round (i.e., $t+1$). For the random variable $Y:=\sum_{i=1}^{b}y_i$, we have $\mathbb{E}[Y]\ge bq^*\ge b/(10\log^2 n)$, where we used $\Pr[y_i=1]=q^*$ and Equation~\eqref{eq-q*}. Since $y_i$'s are independent, applying Chernoff bound (Lemma~\ref{Chernoff} (i) in Appendix~\ref{appendix-ineq}) yields
\begin{equation}\label{chernoff-Y}
\Pr \left[Y\le \frac{b}{20\log^2 n}\right]\le \exp\left(-\Theta\left(\frac{b}{\log^2 n}\right)\right).
\end{equation}
Note that $w_t=\mathcal{O}(s_t/d)=o(s_t)$ implies that $s_t+w_t\le 1.1 s_t$. Furthermore, $1.1 s_t\le N/10$ since $s_t\le C_1N/D=o(N)$. Thus, we can apply Lemma~\ref{constantboundary} for $A= S\cup W$ and the graph obtained from contracting each super node to a node.
Since $|A|=s_t+w_t\le 1.1 s_t\le 1.1 C_1N/D$, we get $2N/5\ge |A|C^{\prime}D$ by selecting $C_1$ to be sufficiently small. Thus, $|\partial(A)|\ge s_tC^{\prime}D$. Furthermore, note that $|\partial(A)|=|\partial(S)\cap U|+|\partial(W)\cap U|=b+|\partial(W)\cap U|$ and $|\partial(W)\cap U|\le w_tD$. Combining the last two statements gives $b\ge s_tC'D-w_tD$. Using $w_t=\mathcal{O}(s_t/d)=o(s_t)$ implies that $b=\Omega(Ds_t)$. Thus, Equation~\eqref{chernoff-Y} implies that w.p. $1-\exp(-\Omega(Ds_t/\log^2 n))=1-\exp(-\Omega(ds_t))$, there will be $\Omega(Ds_t/\log^2 n)=\Omega(ds_t)$ nodes in $U$ which become red in the next round. Note that all such nodes are in different super nodes (see Observation~\ref{obvs1}). If event $\mathcal{E}^*$ holds, then all such super nodes will be strong red in two more rounds. Since $\mathcal{E}^*$ holds w.p. $1-o(1/\log n)$ (as discussed above), there will be $\Omega(ds_t)$ new strong red super nodes after three rounds w.p. $1-\exp(\Omega(-ds_t))-o(1/\log n)$.
\end{proof}

\subsection{Experiments and Real-world Networks}\label{experiments1}
The outcome of our experiments in Figure~\ref{ExperimentsIntext}-(a) are consistent with our theoretical findings. In particular, the rumor does not spread in the flower graph and ER-low (i.e., $p=1/(4\sqrt{n})$) while it does for the moderate expander and ER-high (i.e., $p=4/\sqrt{n}$). Note that in this set-up, a node in the moderate expander is of degree $d+\log^2 n-1 \approx 100$ (actually, we observe in the experiments that for $D=64$ rather than $D = d\cdot 100 = 4\cdot 100$ the rumor already spreads), which indicates the rumor spreads even in very sparse graphs if they possess some level of expansion and community structure. Furthermore, we observe that the process on the moderate expander ends in around $50$ rounds, which indeed appears to be logarithmic rather than linear in $n=16000$ (this is aligned with the bound $\mathcal{O}(\log_d n)$ proven in Theorem~\ref{me-thm}). 

Figure~\ref{ExperimentsIntext}-(b) depicts the extent to which the rumor spreads in Twitter and Facebook graph and random graph model HRG with comparable parameters. (Please refer to Section~\ref{Preliminaries} for more details.) The plots for the other three studied SNs are given in Appendix~\ref{appendix-plots}. We observe that the rumor spreads to a large part of the graph very quickly. This can be explained by the fact that all these graphs have a decent level of expansion and community-like structure, which are the necessary properties for a fast and wide spread according to our theoretical results. As a by-product, our experiments also support that HRG is a suitable choice for modeling real-world SNs.



\section{How to Stop the Rumor Spreading?}\label{how-section}
We present six countermeasures (the first four are inspired by prior work as explained in Section~\ref{prior-work}, but the last two are completely novel) and then compare them. The outcome of our experiments on the countermeasures for Twitter and Facebook graphs and moderate expander are given in Figure~\ref{ExperimentsIntext} and for the other three SNs (T-GE, T-FR, and G+) in Appendix~\ref{appendix-plots}.

\begin{figure*}
\def\tabularxcolumn#1{m{#1}}
\begin{tabularx}{\linewidth}{X}
\begin{tabular}{cc}
\subfloat[]{\includegraphics[width=4.4cm]{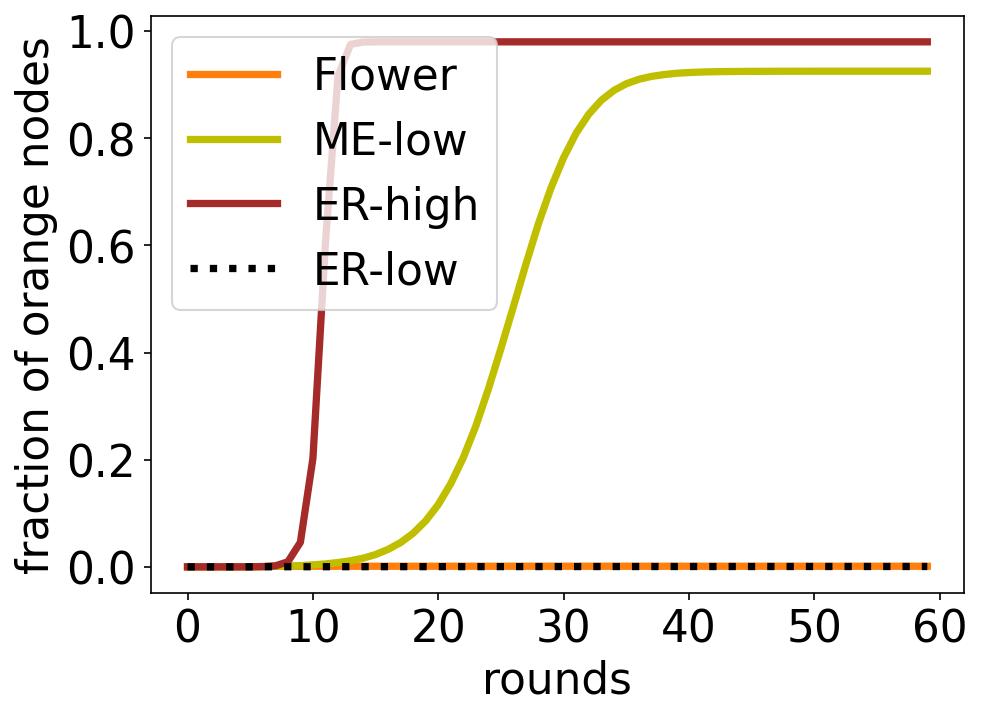}}
    \subfloat[]{\includegraphics[width=4.4cm]{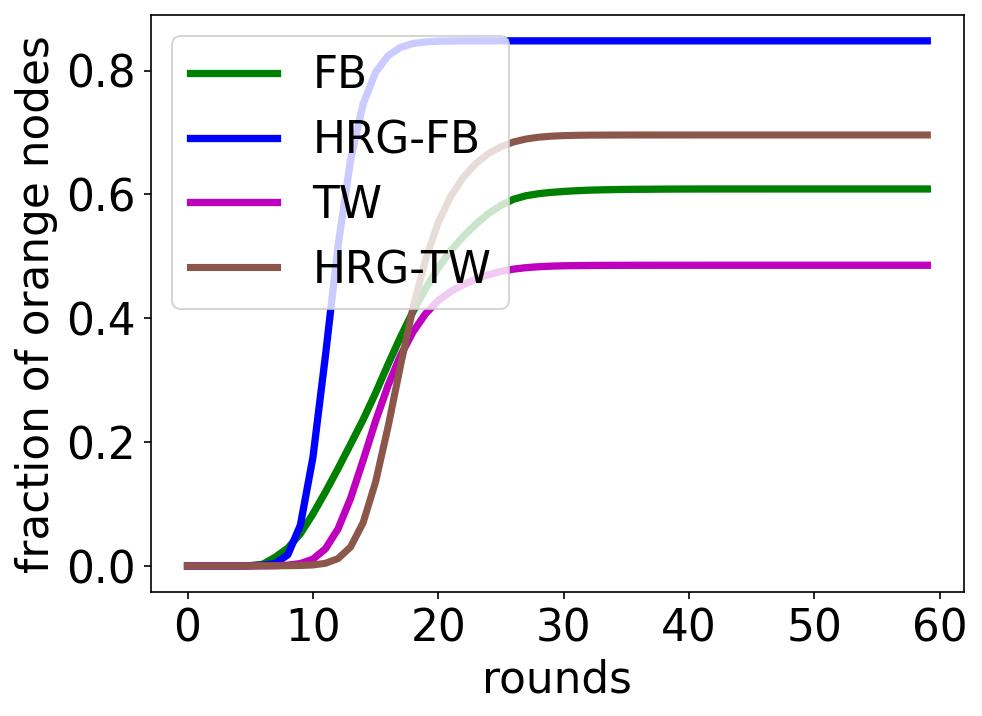}}
    \subfloat[]{\includegraphics[width=4.4cm]{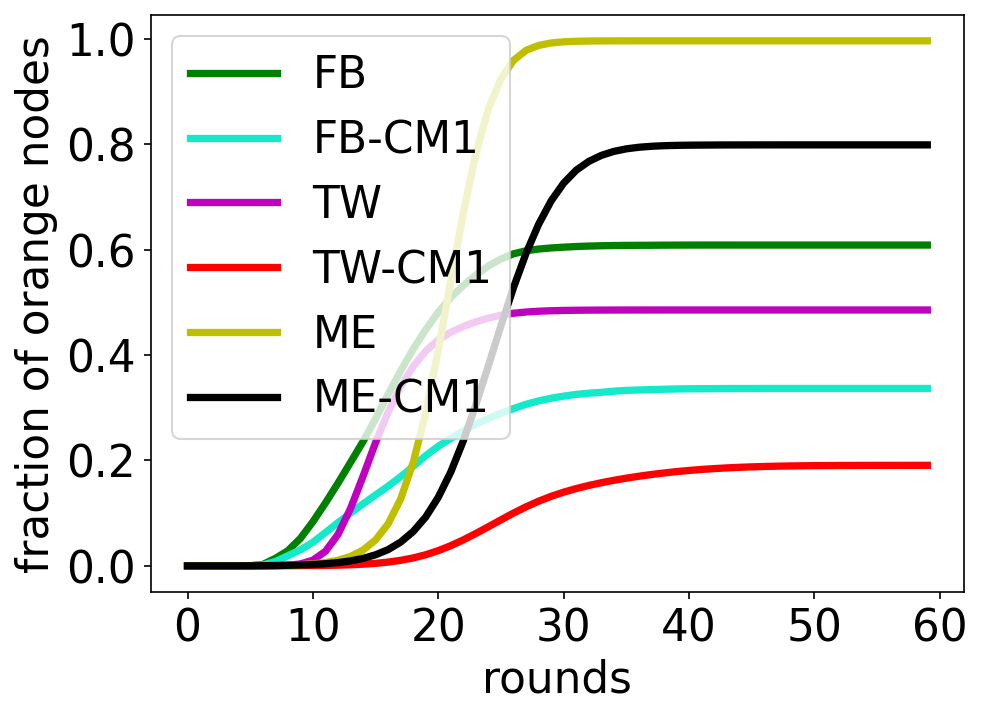}}
\subfloat[]{\includegraphics[width=4.4cm]{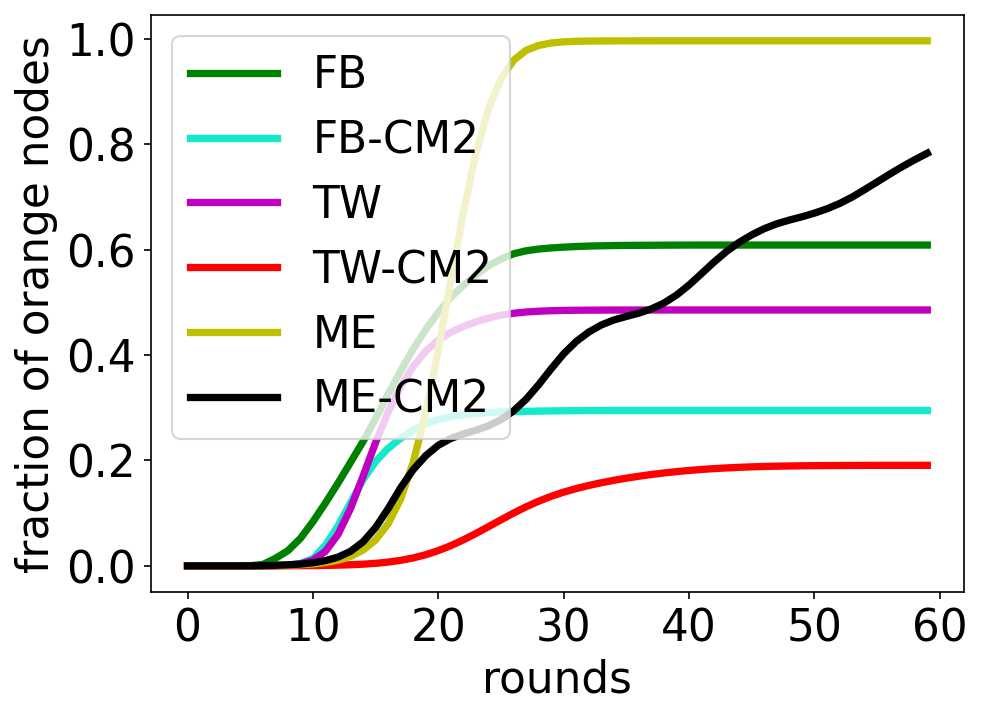}}\\
   \subfloat[]{\includegraphics[width=4.4cm]{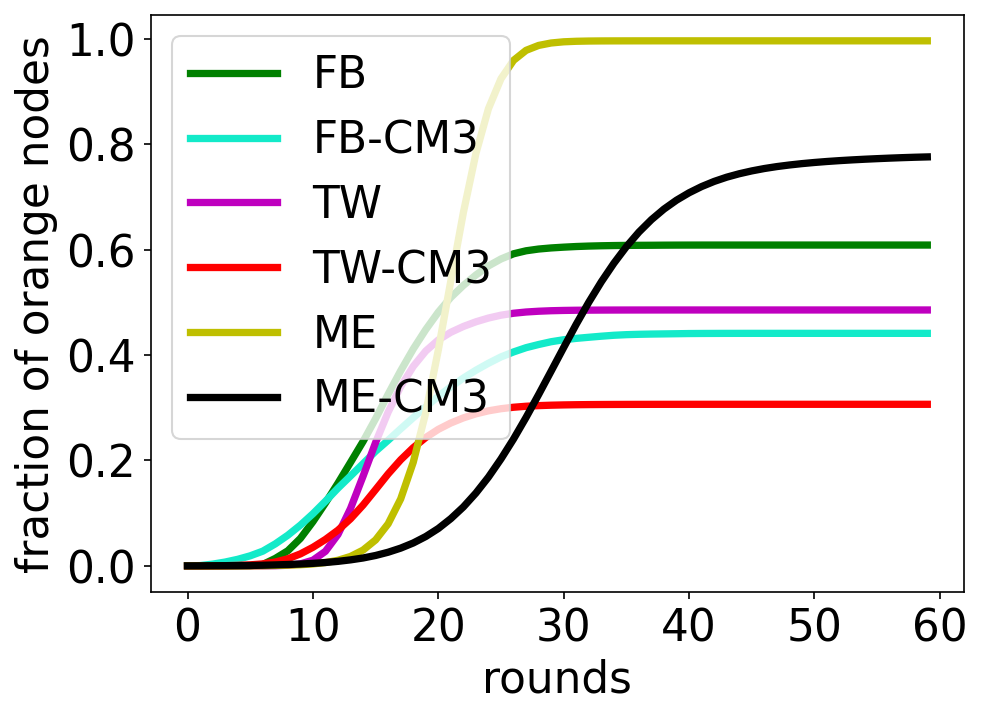}} 
    \subfloat[]{\includegraphics[width=4.4cm]{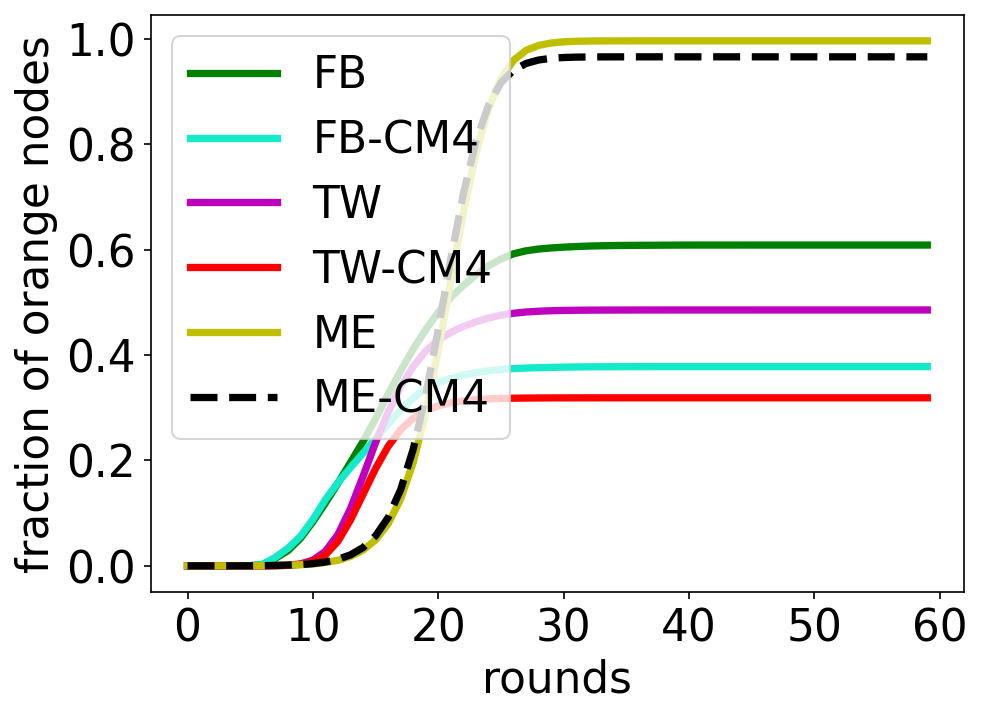}}
    \subfloat[]{\includegraphics[width=4.4cm]{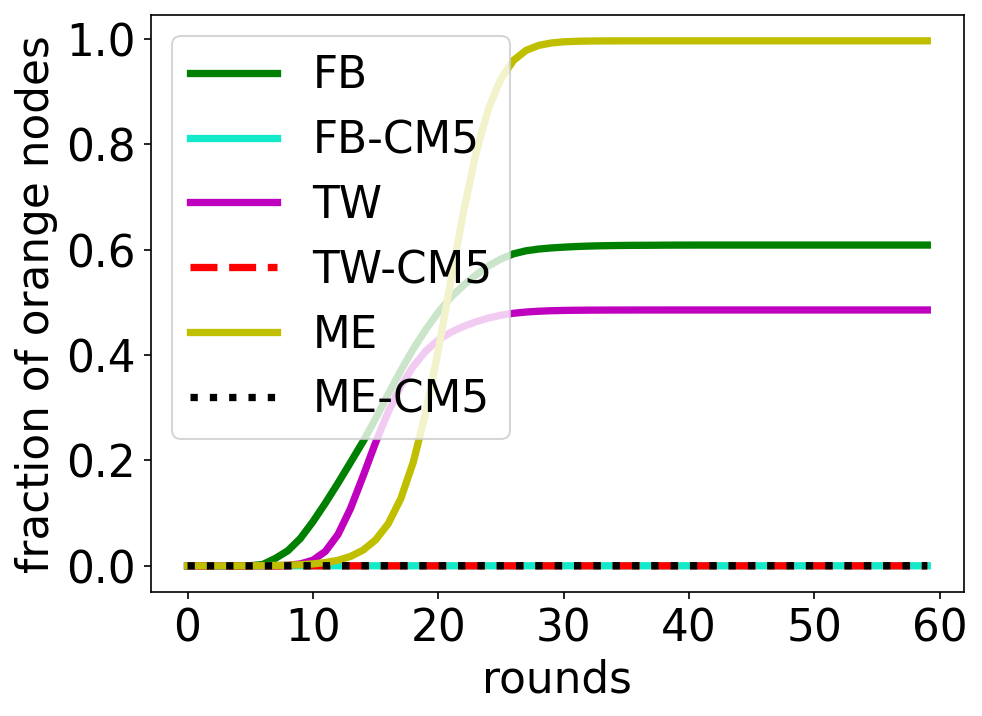}}
    \subfloat[]{\includegraphics[width=4.4cm]{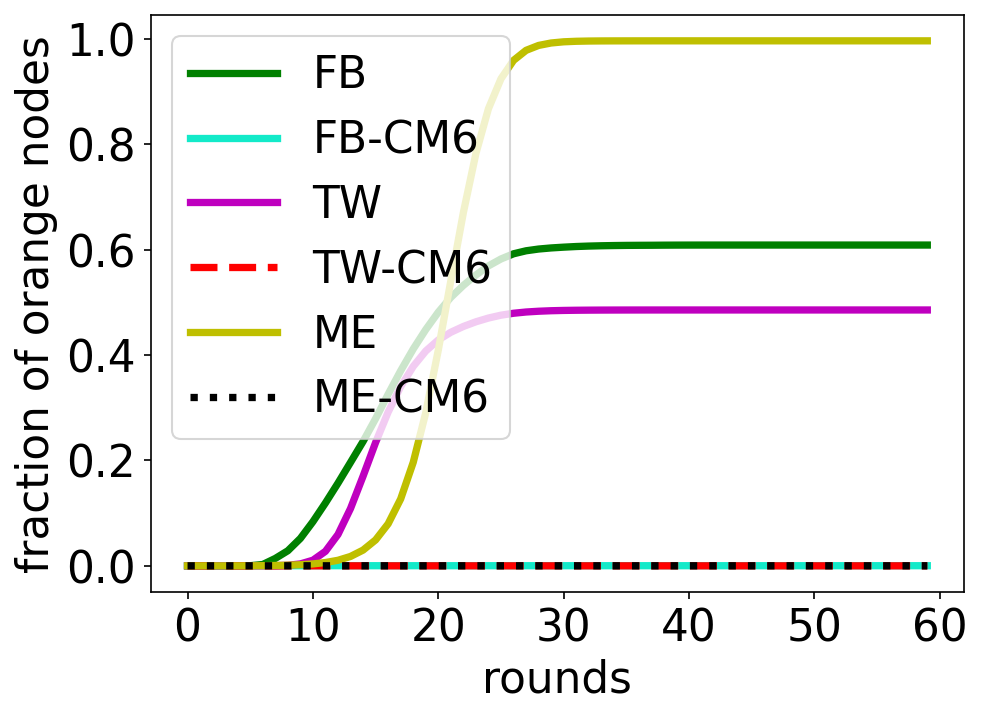}}
\end{tabular}
\end{tabularx}
\vspace{-0.2cm}
\caption{Fraction of orange nodes starting from one randomly chosen red node in (a) $n$-flower, ($n$, $d$)-moderate expander (ME-low) with $d=4$ and $D=64$ and super nodes of size $16$, ER with $p=4/\sqrt{n}$ (ER-high) and $p=1/(4\sqrt{n})$ (ER-low), where $n=16000$ (b) FB and TW graphs and HRG with comparable parameters (c-h) moderate expander (ME) for $n=22000$ and $d=12$ and TW and FB graphs before and after the implementation of countermeasures CM1 to CM6.}\label{ExperimentsIntext}
\vspace{-0.2cm}
\end{figure*}

\paragraph{CM1: Blocking Nodes.}
We assume that the $5\%$ highest degree nodes and $20\%$ randomly chosen nodes are blocked (i.e., do not receive/spread the rumor). As Figure~\ref{ExperimentsIntext}-(c) demonstrates, this countermeasure is not very effective. We believe blocking nodes according to the highest betweenness, closeness, or eigencentrality (instead of highest degree) would not improve the countermeasure significantly since in real-world SNs there is a large overlap between the highest degree nodes and nodes chosen by the mentioned parameters due to certain properties such as the power-law degree distribution.

\paragraph{CM2: Blocking Edges.}
The graph is partitioned into communities using the Louvain algorithm~\cite{Blondel_2008}. In each round of the process, if the fraction of red nodes is above a global threshold $\tau_g$, then we block all the edges which are on the boundary of the ``spreader'' communities. A community is a spreader if its fraction of red nodes is larger than a local threshold $\tau_c$. The blocked edges remain blocked until the community is not a spreader anymore. (Both threshold are set to $0.05$ in our set-up.) Figure~\ref{ExperimentsIntext}-(d) demonstrates while this countermeasure slows down the spread, the rumor still spreads to a large part of the graph. It is worth to mention that around $20-30\%$ of edges were blocked during the process in our experiments. (Unlike other experiments, this was executed only 10 times due to its high computational cost.)

\paragraph{CM3: Accuracy Flags.}
Assume that every time a node is supposed to become red, it rejects the rumor with some \textit{reject} probability $p_r$, and becomes orange directly. In practice, this countermeasure corresponds to for example accuracy flags in online social platforms, in which posts containing certain keywords (say hot controversial or polarizing topics) are automatically accompanied by a banner warning the user about the trustworthiness of the content. 
The outcome of our experiments for $p_r=0.3$, depicted in Figure~\ref{ExperimentsIntext}-(e), demonstrates that the rumor still continues to spread to a significant portion of the community.

\paragraph{CM4: Let's Spread the Truth.}
Let the \textit{truth spreading} process be the same as the rumor spreading with the following two differences: (i) green and light green are used instead of red and orange, respectively (ii) the probability a node becomes green is one half of the probability of becoming red in the rumor spreading process (this is to account for the observation that rumors spread faster than facts, cf.~\cite{vosoughi2018spread}). After $\tau$ rounds into the rumor spreading process, we color an uncolored node green and the truth starts spreading simultaneously. (We assume that the rumor and truth spread only to uncolored nodes, that is, a red/orange node does not become green and vice versa.) The outcome of experiments, depicted in Figure~\ref{ExperimentsIntext}-(f), indicates that this countermeasure cannot stop the rumor effectively even when $\tau=4$ (which implies that there is a strong rumor detection algorithm in place) and the node which starts the truth is the node with the highest degree among the uncolored nodes. We depict the influence of the delay $\tau$ on the final fraction of orange nodes in Figure~\ref{CM4delays} in Appendix~\ref{appendix-delay}.

\paragraph{CM5: Fact Checkers.}
Consider a set of \textit{fact checker} nodes, who starts spreading the truth (i.e., anti-rumor) once exposed to the rumor, as the truth spreading process in CM4. These correspond to ``good citizens'' (e.g., credible news outlets or scientists on the topic) who are educated or incentivized to verify the received information and spread the truth if necessary. 
(In our experiments, we assume they include $10\%$ of the network and are distributed randomly.) This has some similarities to CM4, but instead of starting the spread of the truth by implementing a green node in the graph (which needs to be executed by a third entity), the fact checkers become green and trigger the spread of the truth once contacted by a rumor spreader. Furthermore, the fact checker spread the truth more aggressively: (i) the forgetting parameter $k$ is much larger for the fact checker nodes (say 20 rather than 5) (ii) fact checkers can make their red neighbors green as well (iii) the fact checkers are three times more active in spreading (you can think of each round as three sub-rounds, where all nodes (red/green) spread in the first sub-round while the green fact checker nodes continue to spread in the second and third sub-round too). Note that green nodes which are not fact checker behave as in the original truth spreading process. Our experiments (see Figure~\ref{ExperimentsIntext}-(g)) demonstrate that this countermeasure is very effective.

\paragraph{CM6: Let's Hear It Twice.} We require a node to hear a rumor from at least two of its neighbors before accepting and spreading it (i.e., becoming red), instead of once as in the original process.
Figure~\ref{ExperimentsIntext}-(h) demonstrates that this countermeasure is immensely effective, where in our experiments, initially two randomly chosen nodes are red. We formalize this observation in Theorem~\ref{2timesME}, whose proof is given in Appendix~\ref{appendix-2timesME}.

\begin{theorem}\label{2timesME}
Consider the ($n$,$d$)-moderate expander $\mathcal{M}_{n,d}$ with $d\le n^{\frac{1}{2}-\epsilon}$ for a constant $\epsilon>0$. If initially two super nodes $x$ and $y$, chosen uniformly at random, have red node(s) (and the rest is uncolored) and CM6 is in place, then w.h.p. the rumor does not spread.
\end{theorem}

\paragraph{Comparison of Countermeasures.}
We consider four fundamental criteria that a good countermeasure should possess. To the best of our knowledge, this is the first attempt to formalize such a list of criteria.

\vspace{1mm}

\noindent \textbf{C1: Effective.} A good countermeasure substantially reduces the extent that a rumor spreads.\\
\textbf{C2: Easy To Apply.} An acceptable countermeasure should be feasible and easily executable. If implemented by the agents of the network, it should not require full knowledge of the whole network or the complete history of the process. If it is administrated by a third entity, such as the government, it should not postulate a perfect rumor detection strategy or running algorithms which are computationally very costly.\\
\textbf{C3: Not Against Freedom of Expression.} A countermeasure ideally should not take away the freedom of expression and liberties of the agents.\\
\textbf{C4: Not Too Intrusive.} A countermeasure which demands fundamental changes in the mechanism of information spreading or the network structure is not desirable.\\


\begin{table}[h]
\vspace{-0.2cm}
\begin{tabular}{ |p{0.8cm}||p{0.7cm}|p{0.7cm}|p{0.7cm}|p{0.7cm}|p{2cm}|}
 \hline
  & C1 & C2 & C3 & C4 & Decentralized\\
 \hline
 CM1   & no & jein &  no & no & no\\
  \hline
 CM2   & no & jein &  no & no & no\\
  \hline
 CM3   & no & jein &  yes & yes & no\\
  \hline
 CM4   & no & jein &  yes & yes & jein\\
  \hline
 CM5   & yes & jein &  yes & yes & yes\\
 \hline
 CM6   & yes & yes & yes & yes & yes\\
 \hline
\end{tabular}
\caption{Determining which criteria are satisfied by each countermeasure, where ``jein'' means both yes and no.}
\vspace{-0.2cm}
\label{table}
\end{table}

Table~\ref{table} indicates which criteria each of the proposed countermeasure satisfies. Note that it is inherently difficult to measure the above criteria in a strict quantitative manner. Thus, the entries in the table are relative and up to interpretation. The choices for C1 are according to the results depicted in Figure~\ref{ExperimentsIntext}. The entries for C2 are mostly set to jein since while they are not extremely difficult to implement, they need a smart rumor detection strategy or the full knowledge of the network. Furthermore, CM1 and CM2 violate C3 since they clearly intrude the freedom of expression and do not satisfy C4 since they change the network structure radically. The other countermeasures, arguably, satisfy the last two criteria. Please refer to Appendix~\ref{appendix-table} for a more comprehensive discussion on the entries of Table~\ref{table}.

We say a countermeasure is \textit{decentralized} if it is executed by the members of the network rather than being enforced by a third party such as the government or an online social platform management team. Summarizing the entries of Table~\ref{table} implies that, interestingly, the decentralized countermeasures, namely CM5 and CM6 (and CM4, up to some degree), satisfy most of the desired criteria while the centralized ones do not. Hence, instead of developing centralized countermeasures which need to be imposed by a forceful third entity, the focus should be devoted to the design and implementation of decentralized countermeasures which can be obtained through educating the members of the network. In short, educating is preferred over regulating.

\vspace{-0.3cm}
\section{Conclusion}
We introduced a rich rumor spreading model and building on our theoretical and experimental findings, we argued that the abundance of community structures and good expansion properties are two of the main driving forces behind the spread of rumors. A potential avenue for future research is to determine other graph parameters which govern the spread of rumors. We also investigated several countermeasures. We observed that the decentralized countermeasures (which do not require a direct and forceful interference of a third entity but rather the education of the network's members) outperform the centralized ones vigorously. Therefore, a natural suggestion for the future studies is the shift of focus from centralized countermeasures to decentralized ones, which have been scarcely investigated by the prior work.

\newpage 

\bibliographystyle{named}
\bibliography{ref}

\newpage

\appendix

\section{Visualization of Process on Facebook Network}
\label{appendix-visualize}

Figure~\ref{visualize-fig} visualizes the spread of the rumor on the Facebook SN starting from a randomly selected red node.

\FloatBarrier
\begin{figure*}[p]
\begin{tabular}{cccc}
\subfloat[Day 0]{\includegraphics[width=4cm, height=5cm]{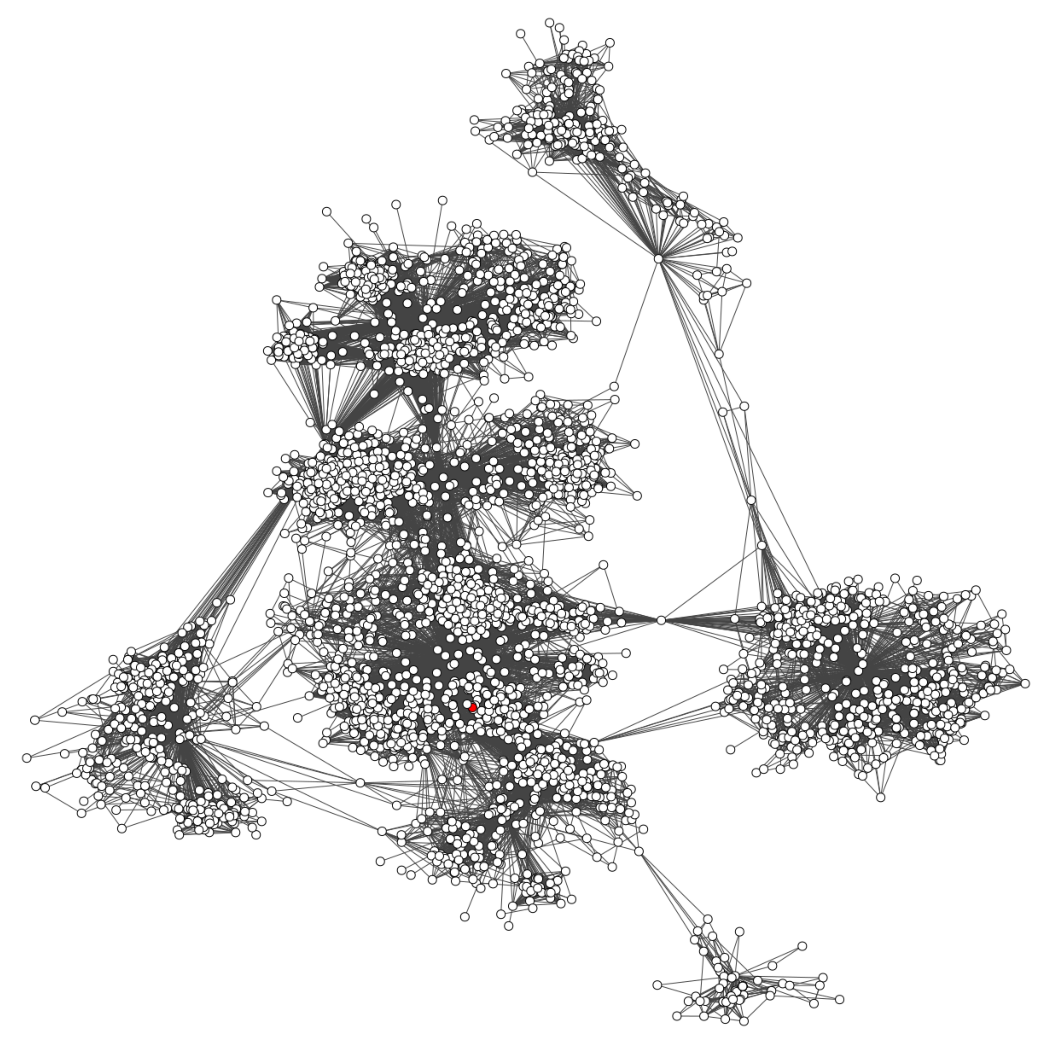}} 
    \subfloat[Day 2]{\includegraphics[width=4cm, height=5cm]{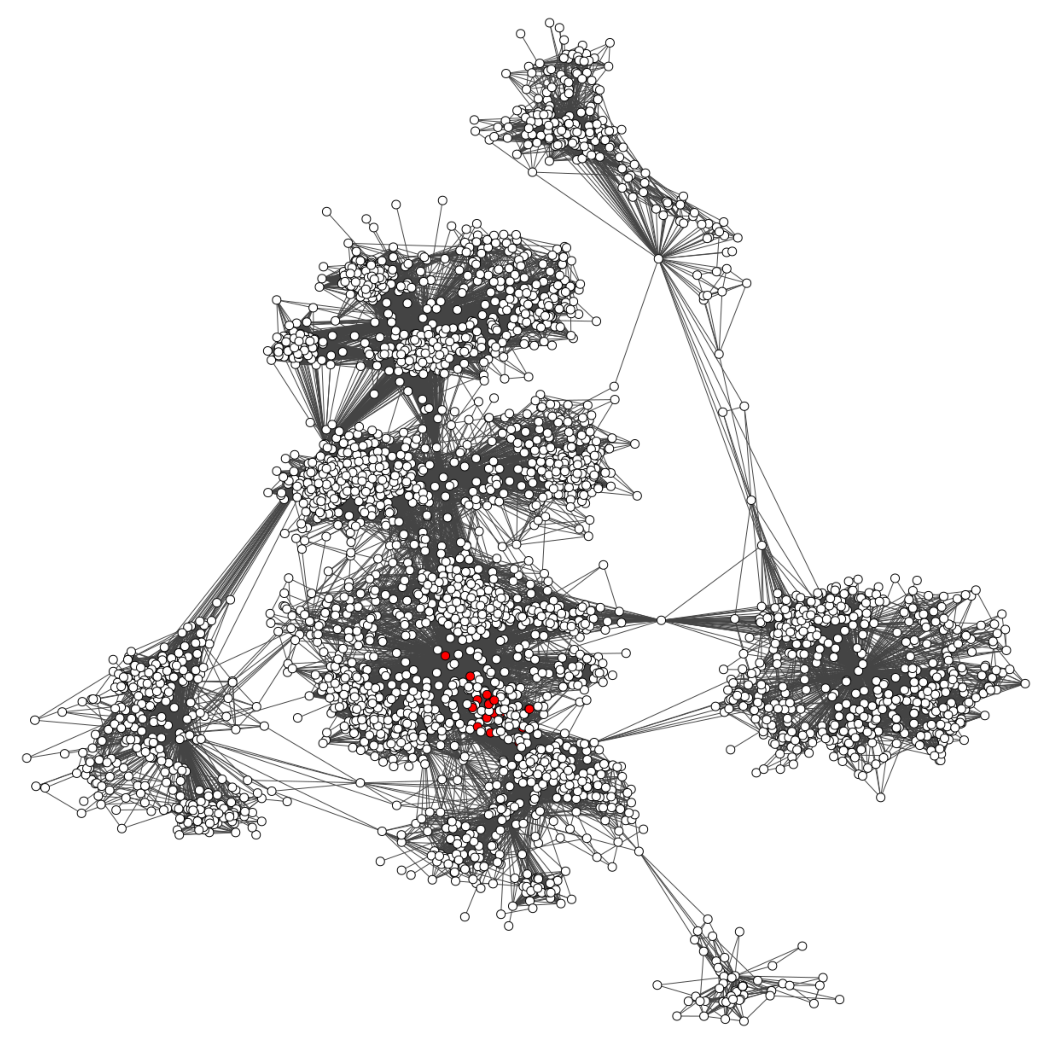}\label{day1}}
    \subfloat[Day 4]{\includegraphics[width=4cm, height=5cm]{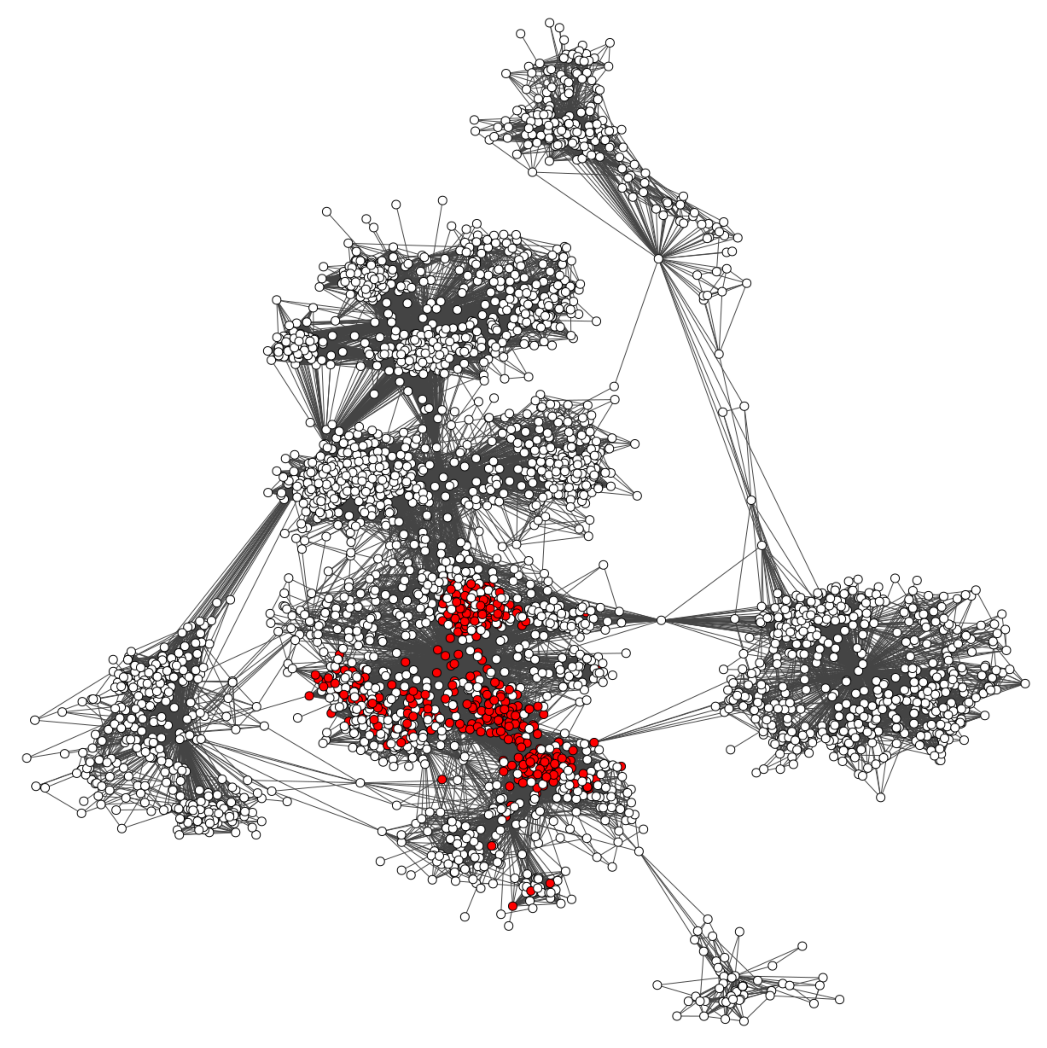}}
\subfloat[Day 6]{\includegraphics[width=4cm, height=5cm]{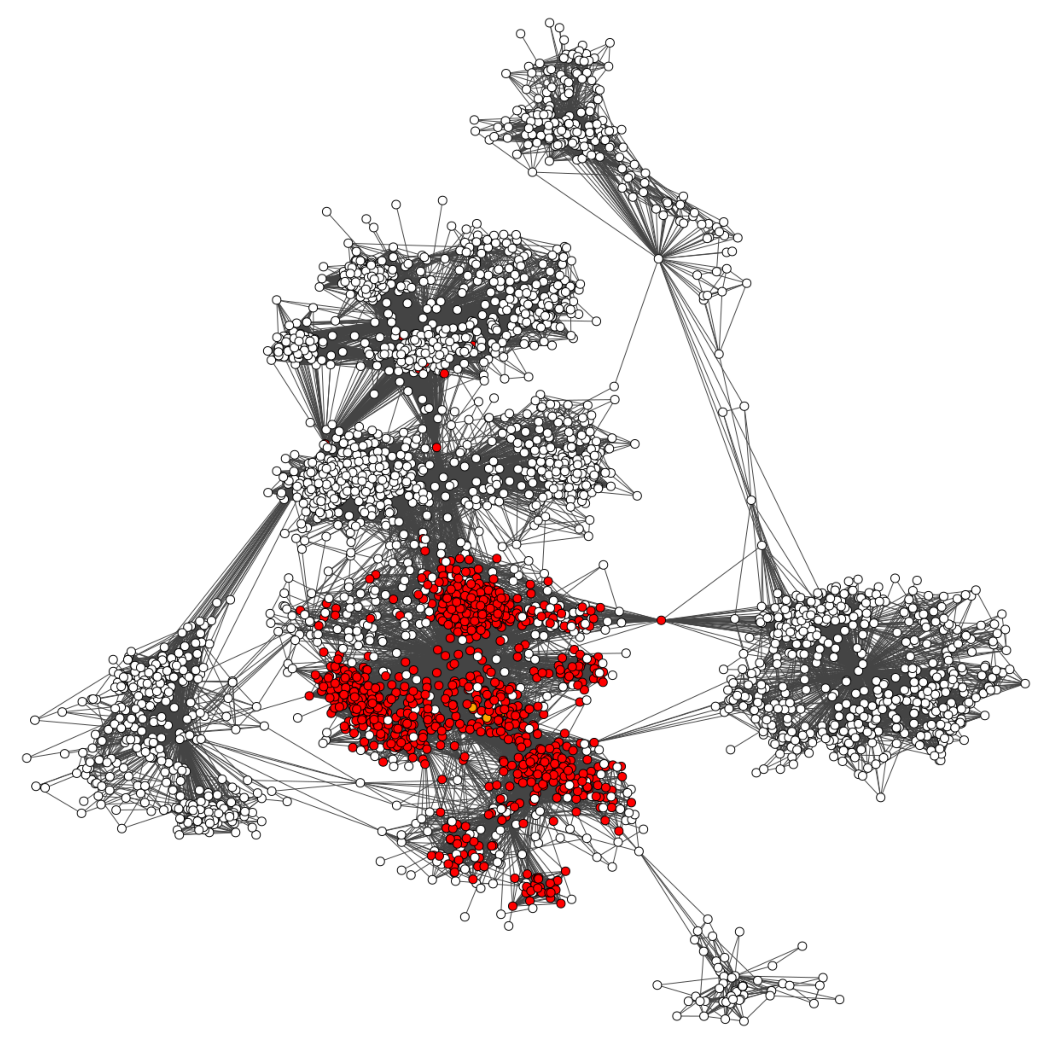}}\\
   \subfloat[Day 8]{\includegraphics[width=4cm, height=5cm]{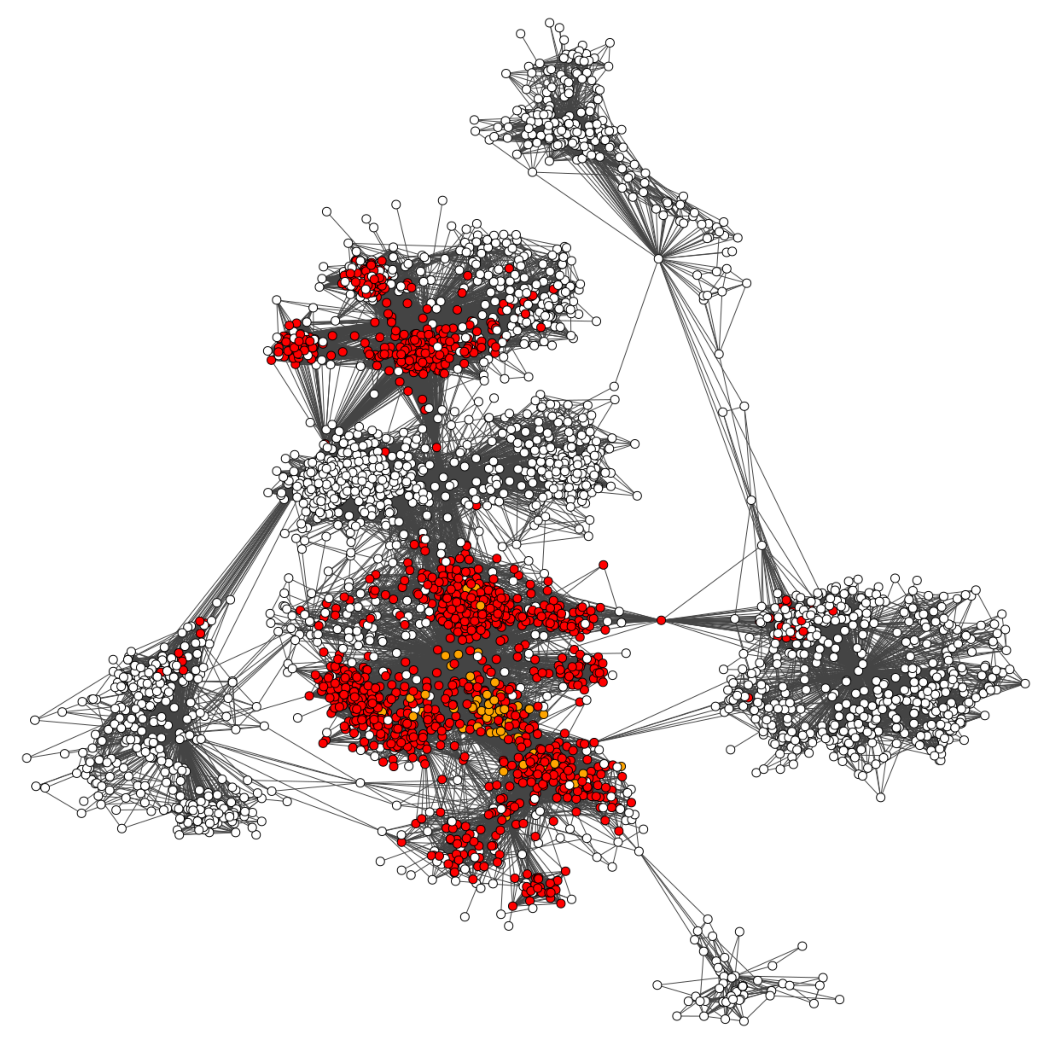}} 
    \subfloat[Day 10]{\includegraphics[width=4cm, height=5cm]{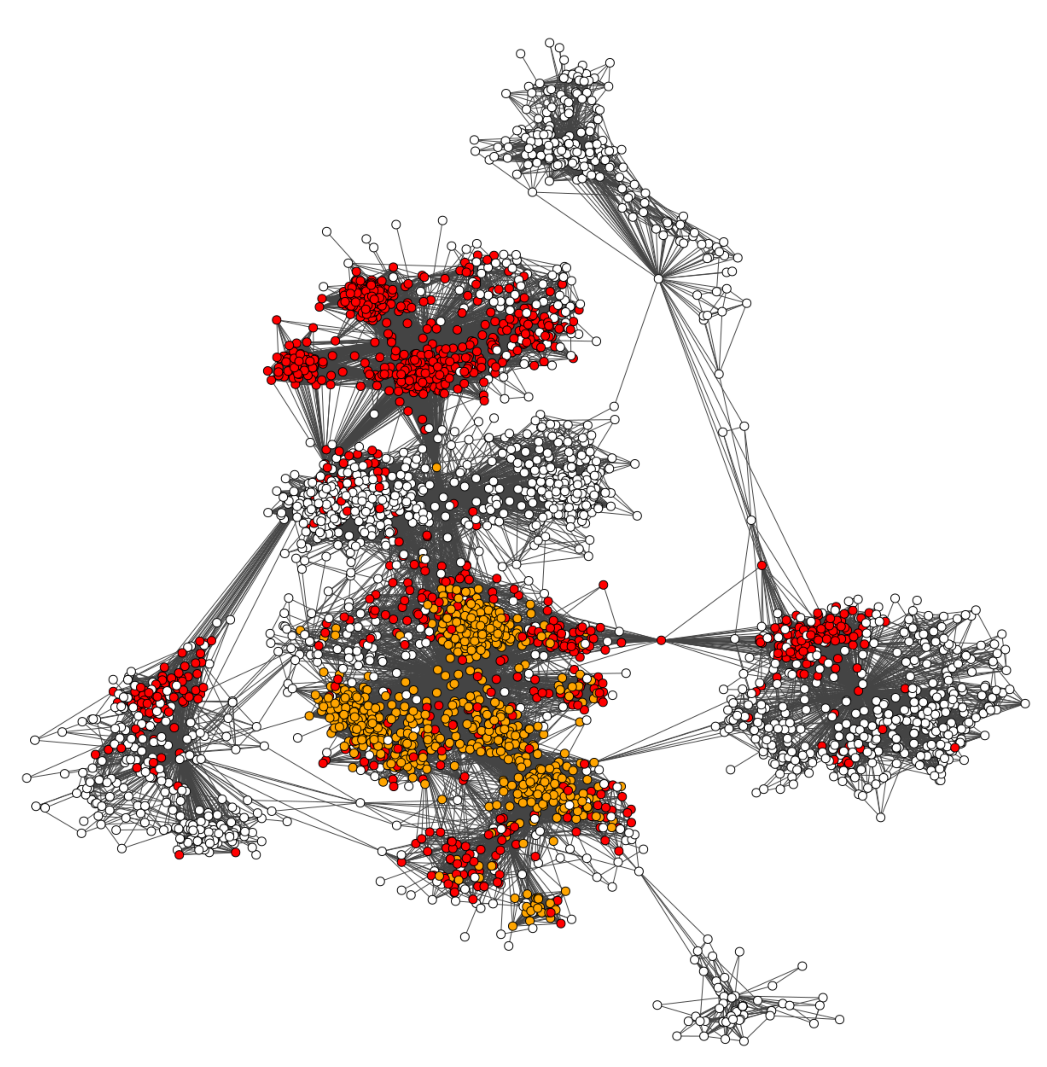}}
    \subfloat[Day 12]{\includegraphics[width=4cm, height=5cm]{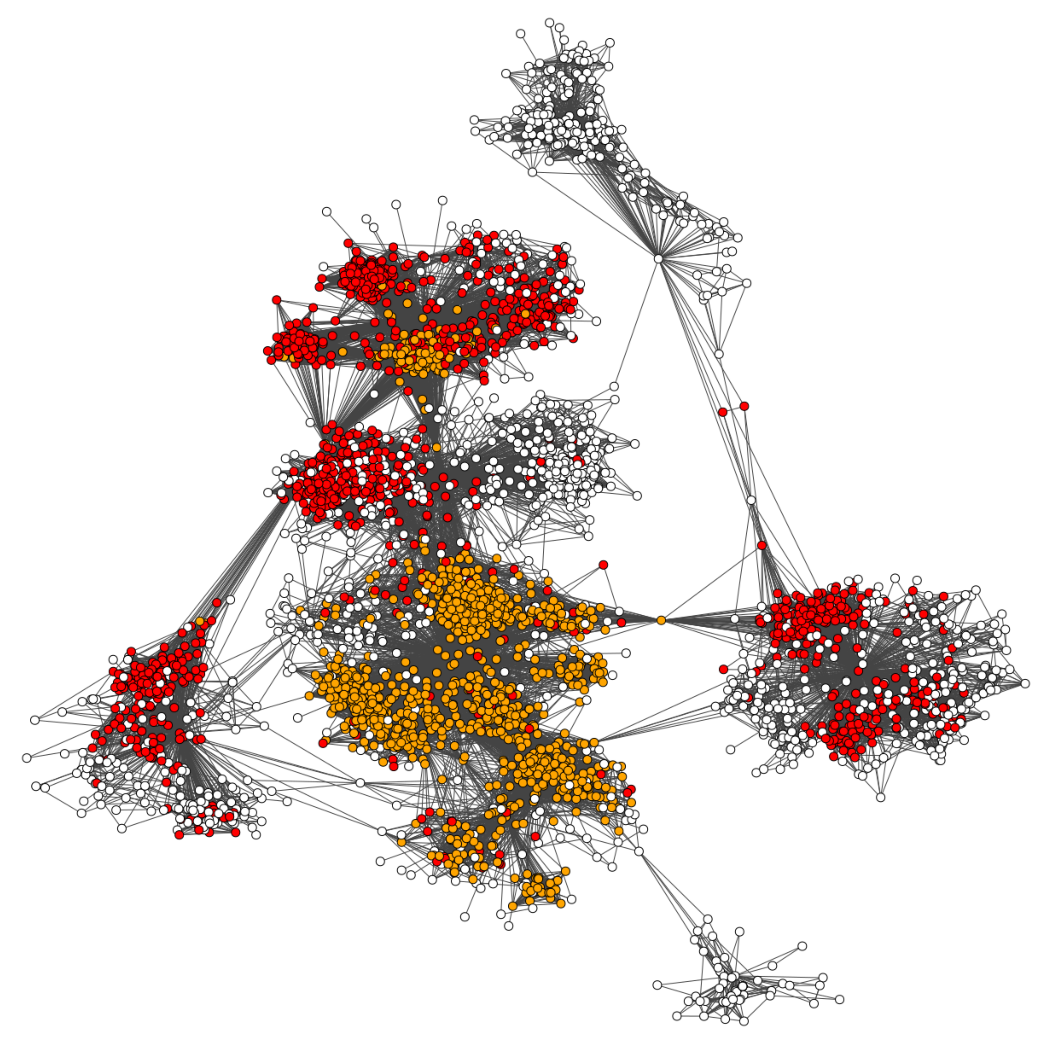}}
    \subfloat[Day 14]{\includegraphics[width=4cm, height=5cm]{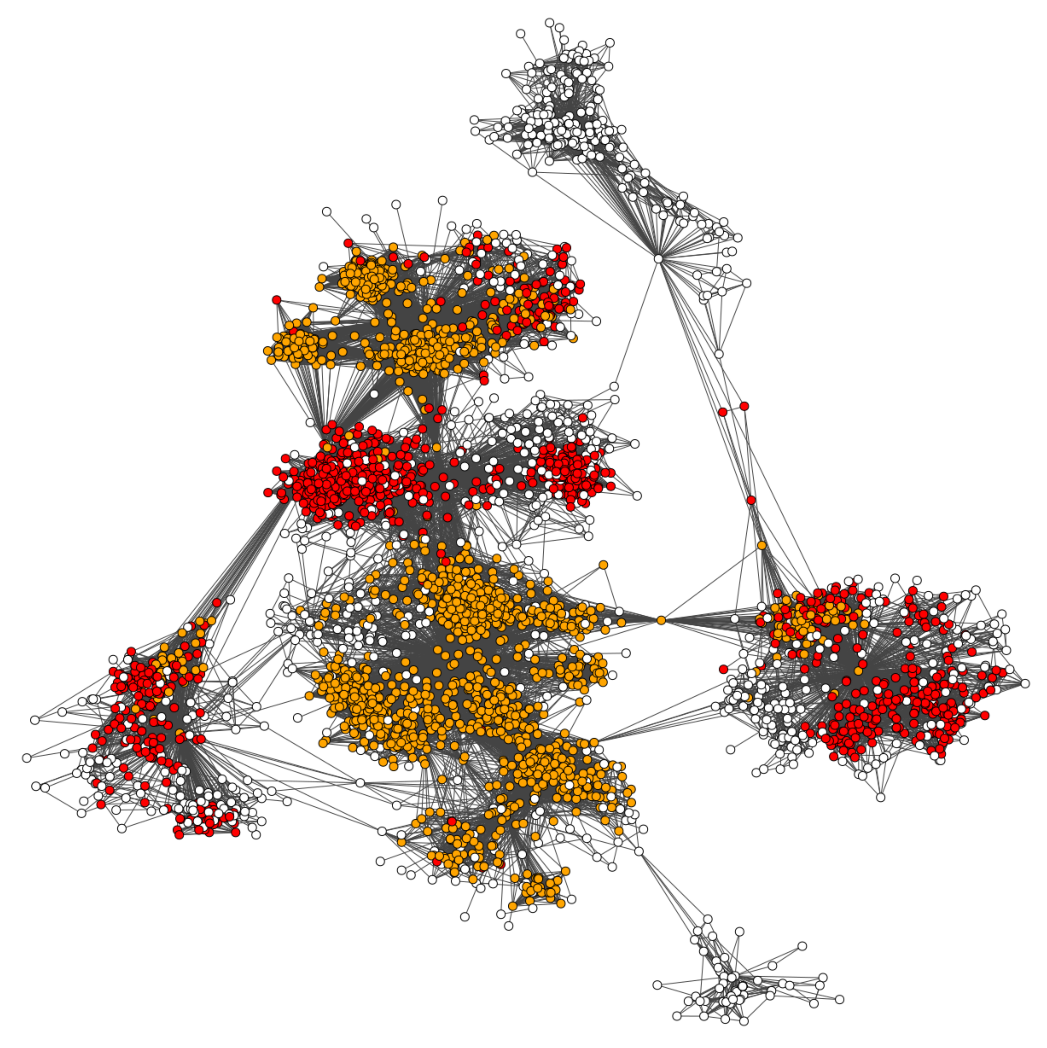}}\\
    \subfloat[Day 16]{\includegraphics[width=4cm, height=5cm]{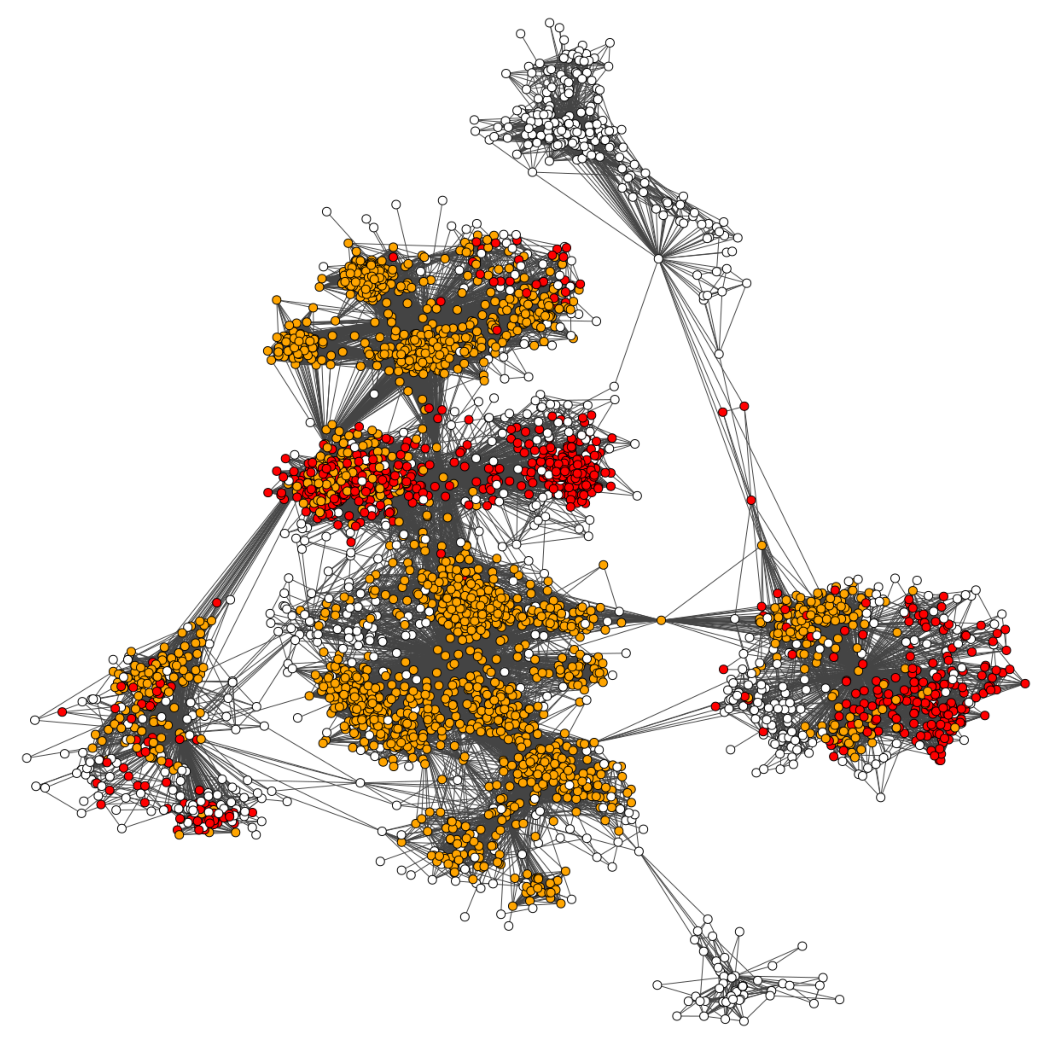}}
    \subfloat[Day 18]{\includegraphics[width=4cm, height=5cm]{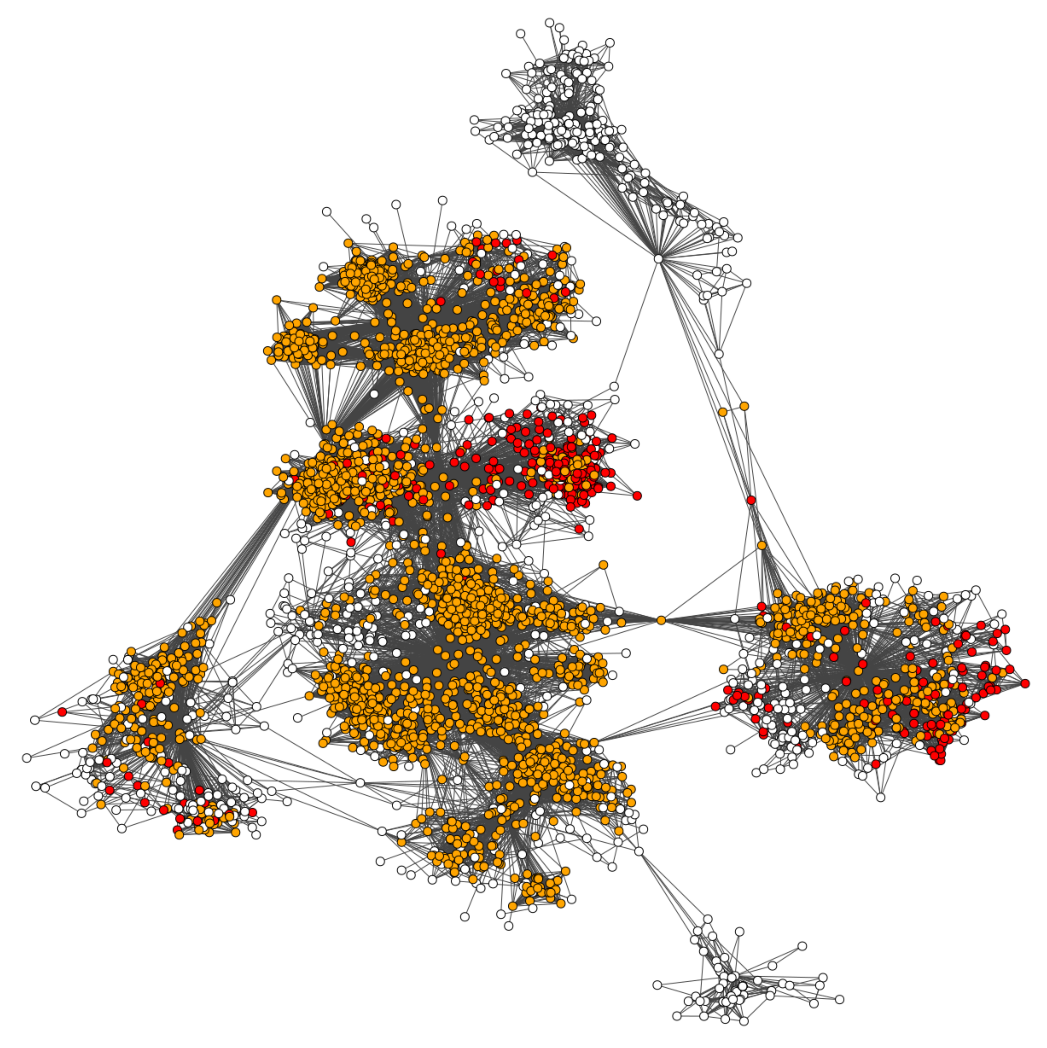}}
    \subfloat[Day 20]{\includegraphics[width=4cm, height=5cm]{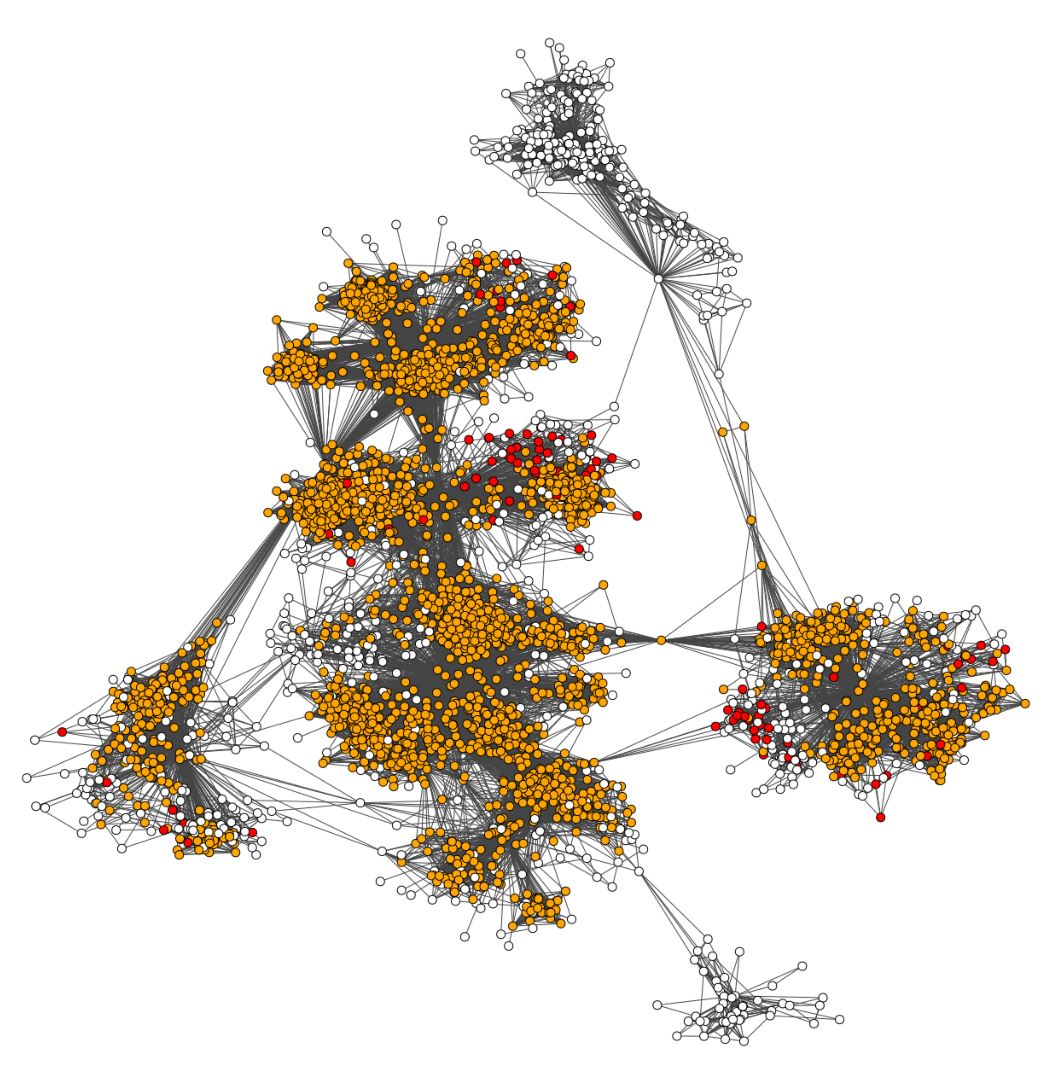}} 
    \subfloat[Day 22]{\includegraphics[width=4cm, height=5cm]{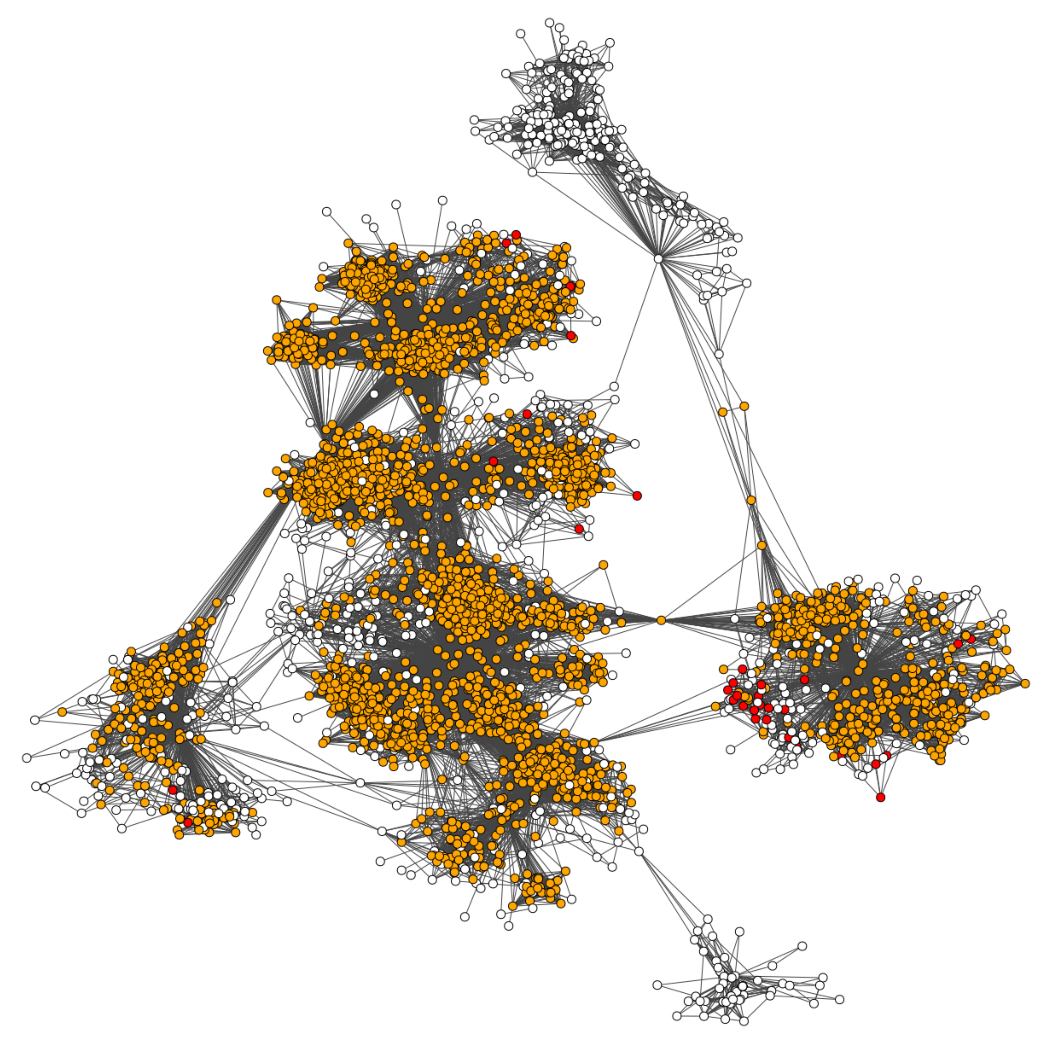}}\\
    \subfloat[Day 24]{\includegraphics[width=4cm, height=5cm]{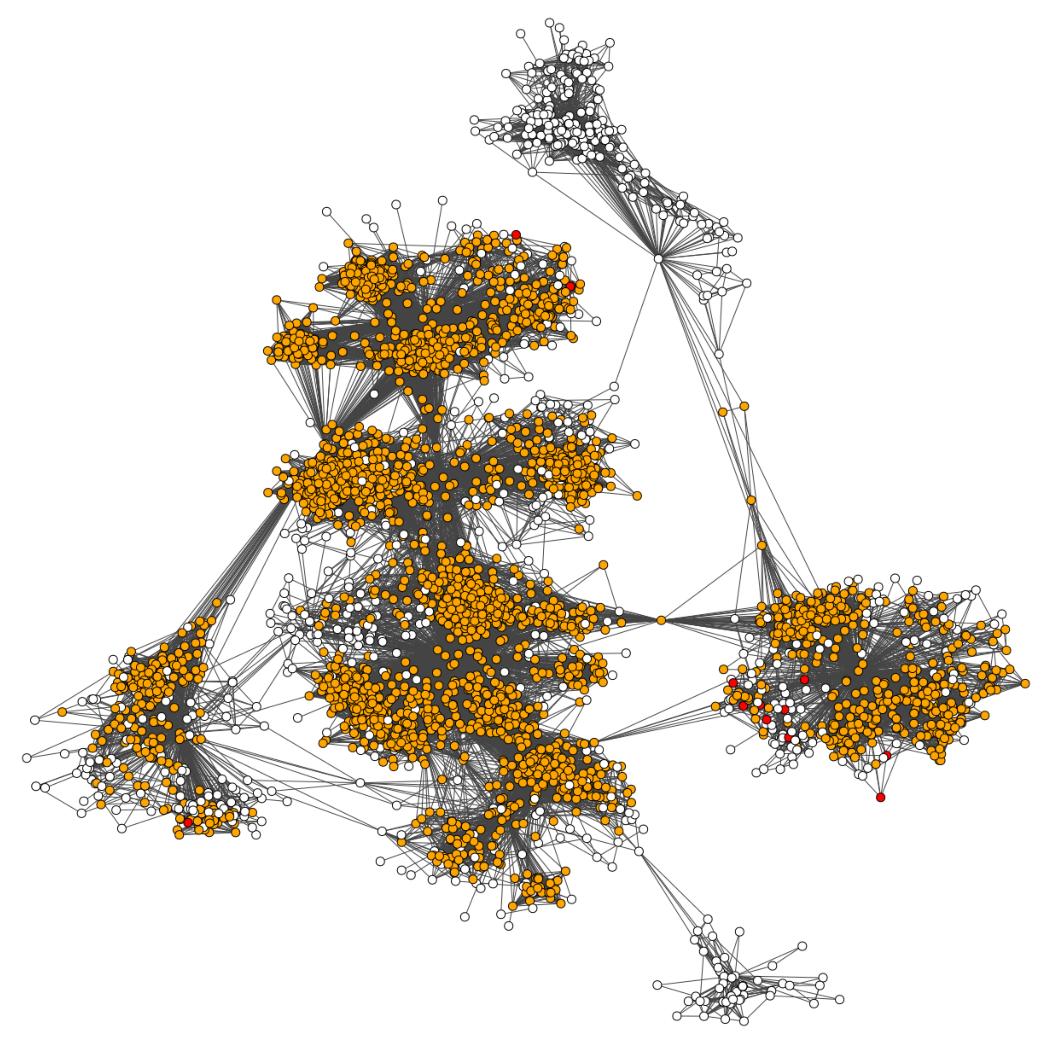}}
    \subfloat[Day 27]{\includegraphics[width=4cm, height=5cm]{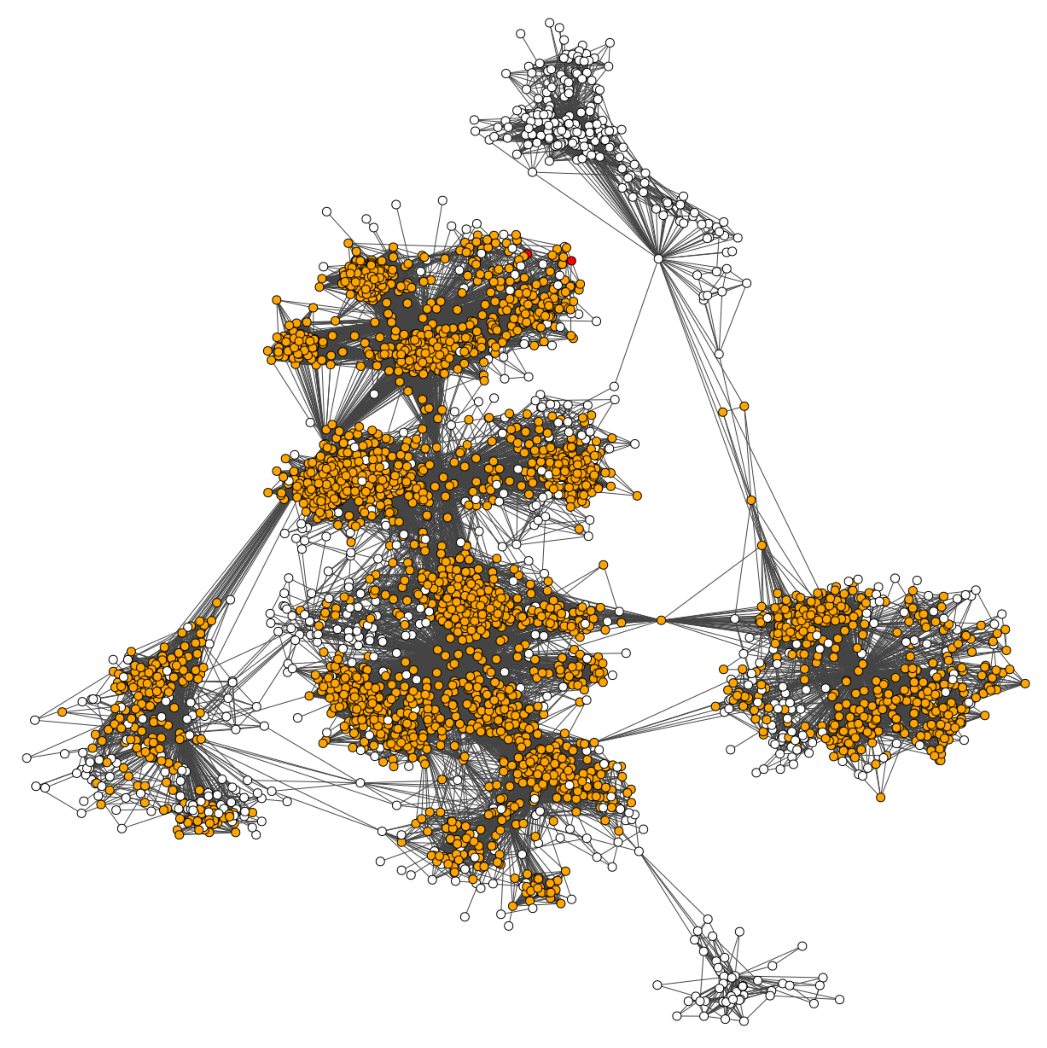}}
    \subfloat[Day 35]{\includegraphics[width=4cm, height=5cm]{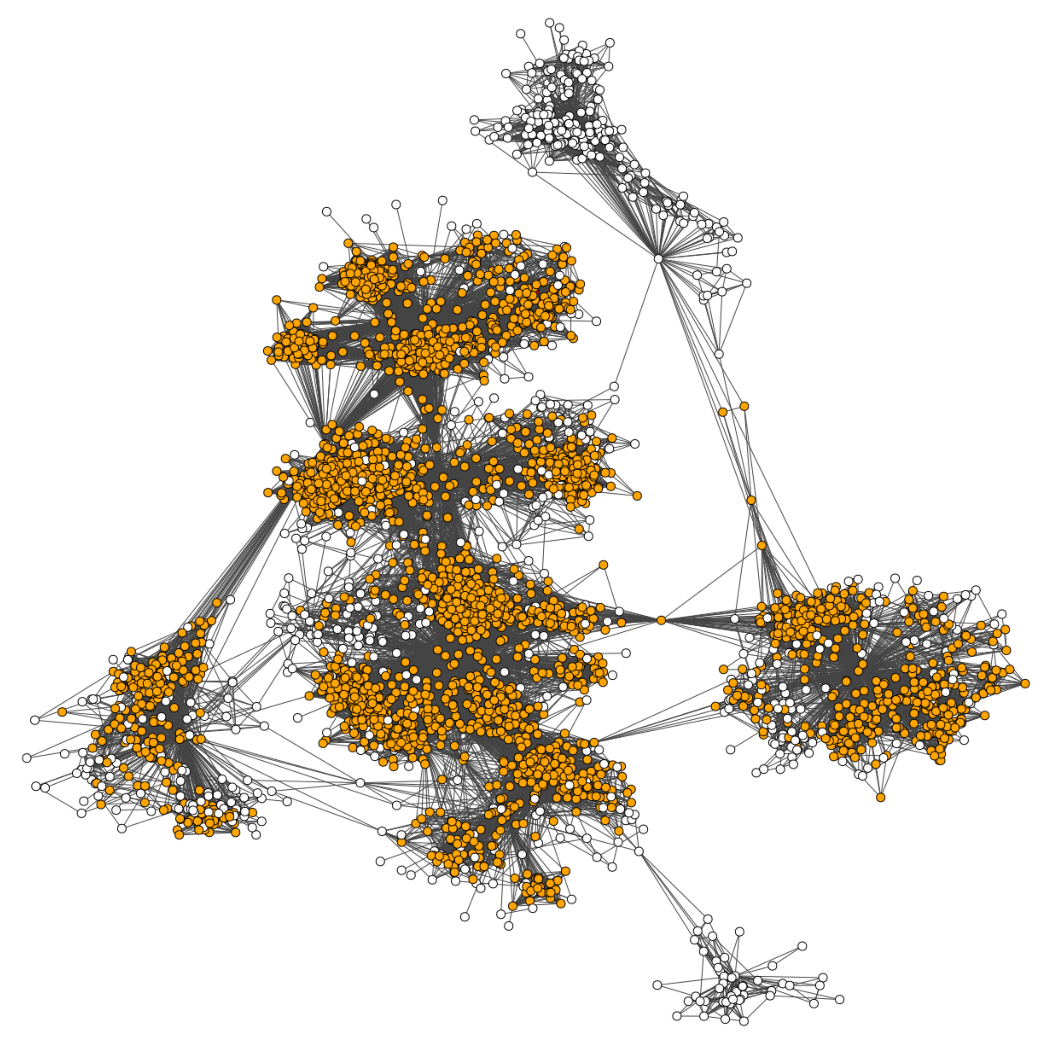}}
    \subfloat[Day 42]{\includegraphics[width=4cm, height=5cm]{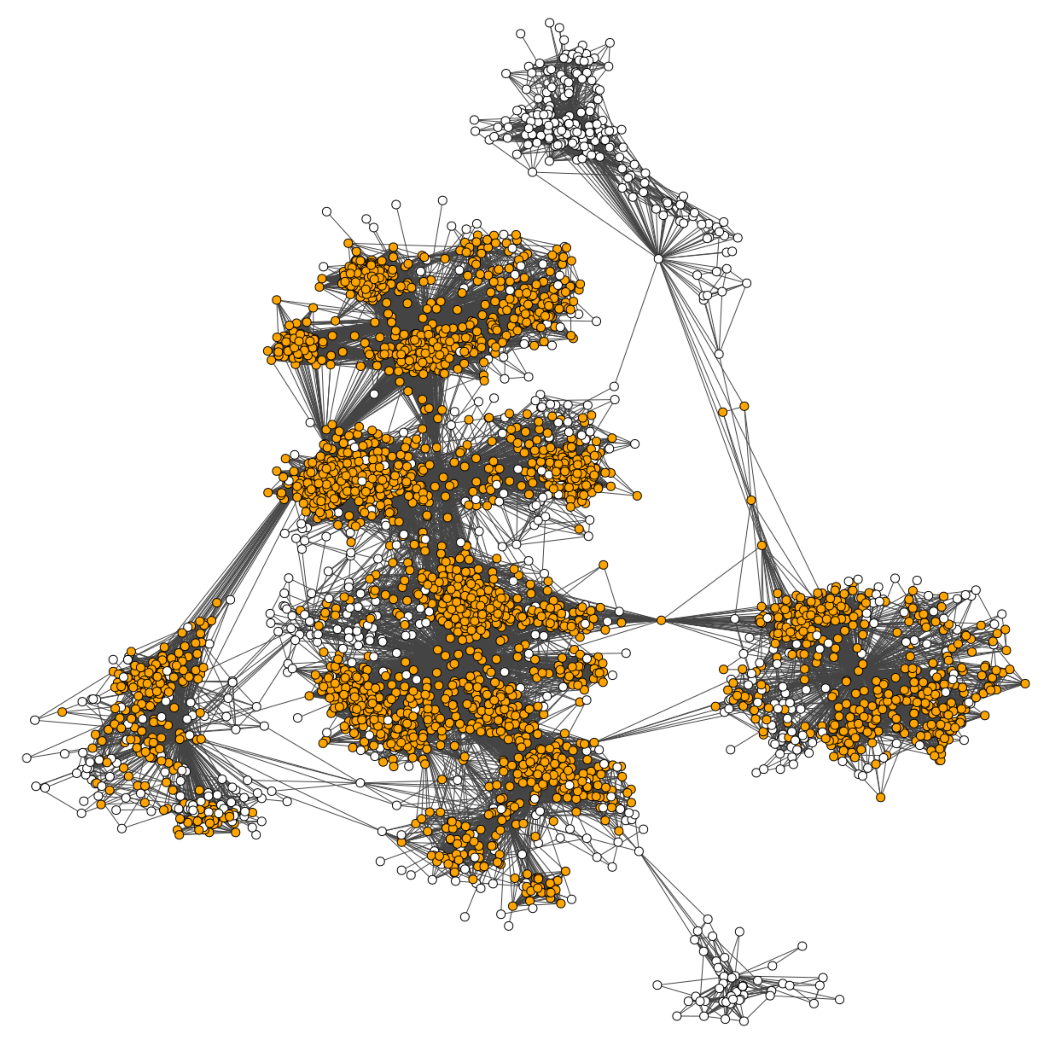}}\\
\end{tabular}

\caption{The visualization of spreading the rumor in the FB SN with 1 randomly chosen initial red node. The process ends after 42 rounds with almost $80\%$ of the nodes being orange.}
\label{visualize-fig}
\end{figure*}

\section{Inequalities}\label{appendix-ineq}

Here, we provide two standard probabilistic tools, Chernoff bound and Markov's inequality.

\begin{lemma}[Chernoff bound, cf.~\cite{dubhashi2009concentration}]
\label{Chernoff}
Suppose that $x_1,\cdots,x_n$ are independent Bernoulli random variables and let $X$ denote their sum, then for $0\leq \delta\leq 1$
\begin{itemize}
\item[(i)]$\textrm{Pr}[X\leq \left(1-\delta\right)\mathbb{E}[X]]\leq \exp\left({-\frac{\delta^2\mathbb{E}[X]}{2}}\right)$
\item[(ii)]$\textrm{Pr}[\left(1+\delta\right)\mathbb{E}[X]\leq X]\leq \exp\left({-\frac{\delta^2\mathbb{E}[X]}{3}}\right)$.
\end{itemize}
\end{lemma}

\begin{lemma}[Markov’s inequality, cf.~\cite{dubhashi2009concentration}]
\label{Markov}
Let $X$ be a non-negative random variable with finite expectation and $a>0$, then
\[
\Pr[X\ge a]\le \frac{\mathbb{E}[X]}{a}.
\]
\end{lemma}
\section{Tightness of Theorem~\ref{er-thm}}
\label{appendix-ER-tightness}

Here, we argue that the conditions of Theorem~\ref{er-thm} cannot be relaxed. In Theorem~\ref{er-tight}, we prove that on $\mathcal{G}_{n,p}$ if $p \geq 1/n^{(1/2) - \epsilon}$ for some $\epsilon>0$, then after one round there will be $\Omega(n^{2\epsilon})$ red nodes w.h.p. This implies that the bound on $p$ in Theorem~\ref{er-thm} is tight. Afterward, we show that if we replace ``constant probability'' with ``w.h.p.'' in Theorem~\ref{er-thm}, the statement of the theorem is no longer true.

\begin{theorem}\label{er-tight}
Consider the coloring where only one node is red (and the rest is uncolored) on $\mathcal{G}_{n,p}$ with $p \geq 1/n^{(1/2) - \epsilon}$ for an arbitrary constant $\epsilon > 0$. Then, after one round there are $\Omega(n^{2\epsilon})$ red nodes w.h.p.
\end{theorem}

\begin{proof}
Let us first define two events and bound their probability.
\begin{itemize}
    \item $\mathcal{A}:=$ The event that there is no node $v$ such that $d(v)< \frac{1}{2}(n-1)p$ or $d(v)> 2(n-1)p$.
    \item $\mathcal{B}:=$ The event that there are no two distinct nodes $v,u$ such that $|N(v)\cap N(u)|\le\frac{1}{2}(n-2)p^2$.
\end{itemize}

Consider an arbitrary node $v$. Label the other nodes from $u_1$ to $u_{n-1}$. Let Bernoulli random variable $x_i$, for $1\le i\le n-1$, be 1 if and only if the edge $\{v,u_i\}$ is present. Note that $d(v)=\sum_{i=1}^{n-1}x_i$ and $\mathbb{E}[d(v)]=(n-1)p$. Since $x_i$'s are independent, using Chernoff bound (Lemma~\ref{Chernoff} in Section~\ref{appendix-ineq}) gives
\begin{align*}
\Pr & \left[ \frac{1}{2}(n-1)p \le  d(v)\le 2(n-1)p \right]
\\ \ge & 1-\exp\left(-\frac{(n-1)p}{8}\right)-\exp\left(-\frac{(n-1)p}{3}\right)
\\ \ge & 1-\exp\left(-\frac{(n-1)}{8n^{(1/2)-\epsilon}}\right)-\exp\left(-\frac{(n-1)}{3n^{(1/2)-\epsilon}}\right)
\\ \ge & 1-\exp\left(-\Theta(\sqrt{n})\right).
\end{align*}
Since we have $n$ nodes, we get
\begin{equation}
    \label{event-A}
    \Pr[\mathcal{A}]\ge 1- n \exp\left(-\Theta(\sqrt{n})\right)=1-o(1).
\end{equation}
Now, we bound the probability of event $\mathcal{B}$. Consider two arbitrary distinct nodes $v$ and $u$. Label the remaining nodes from $w_1$ to $w_{n-2}$. Define Bernoulli random variable $y_i$, for $1\le i\le n-2$, to be 1 if and only if the edges $\{v,w_i\}$ and $\{u,w_i\}$ are present. Note that $|N(v)\cap N(u)|=\sum_{i=1}^{n-2}y_i$ and $\mathbb{E}[|N(v)\cap N(u)|]=(n-2)p^2$. Since $y_i$'s are independent, using Chernoff bound (Lemma~\ref{Chernoff} in Section~\ref{appendix-ineq}) yields

\begin{align*}
    \Pr\left[\left|N(v)\cap N(u)\right|\le \frac{1}{2}(n-2)p^2\right]\le & \exp\left(-\frac{(n-2)p^2}{8}\right) \\ \le & \exp\left(\frac{n-2}{8n^{1-2\epsilon}}\right) \\ = & \exp\left(-\Theta\left(n^{2\epsilon}\right)\right).
\end{align*}

Since there are ${n \choose 2}$ ways to select two distinct nodes, we have
\begin{equation}
    \label{event-B}
    \Pr[\mathcal{B}]\ge 1-{n \choose 2}\cdot \exp\left(-\Theta\left(n^{2\epsilon}\right)\right)=1-o(1).
\end{equation}

Let $v$ be the node which is red in $\mathcal{C}_0$ and $u$ be a neighbor of $v$. We are interested in the probability of $\mathcal{C}_1(u)=r$ conditioning on $\mathcal{A}$ and $\mathcal{B}$. Note that event $\mathcal{A}$ implies that $|N(v)\cup N(u)|\le 4(n-1)p$ and event $\mathcal{B}$ asserts that $|N(v)\cap N(u)|\ge \frac{1}{2}(n-2)p^2$. Therefore, we have
\[
\Pr[\mathcal{C}_1=r|\mathcal{A}\land \mathcal{B}]\ge \frac{(1/2)(n-2)p^2}{2\cdot 4(n-1)p}\ge \frac{p}{32}.
\]
Let $Z$ denote the number of nodes that $v$ makes red in the first round. Since $\mathcal{A}$ implies that $d(v)\ge \frac{1}{2}(n-1)p$, we get
\[
\mathbb{E}[Z|\mathcal{A}\land \mathcal{B}]\ge \frac{1}{2}(n-1)p\cdot \frac{p}{32}\ge \frac{n-1}{64}\cdot \frac{1}{n^{1-2\epsilon}}\ge \frac{n^{2\epsilon}}{128}
\]
where we used $p\ge 1/n^{(1/2)-\epsilon}$ and $n-1\ge n/2$. Now, applying Chernoff bound yields
\begin{equation}
    \label{z-bound}
    \begin{split}
           \Pr\left[Z\le \frac{1}{2}\cdot\frac{n^{2\epsilon}}{128}\Big|\mathcal{A}\land \mathcal{B}\right] \le \exp\left(-\frac{n^{2\epsilon}}{4\cdot 128}\right)  = \\ \exp\left(-\Theta\left(n^{2\epsilon}\right)\right) 
    \end{split}
\end{equation}

Now, combining
\[
\Pr\left[Z> \frac{n^{2\epsilon}}{256}\right]\ge \Pr\left[Z > \frac{n^{2\epsilon}}{256}\Big|\mathcal{A}\land \mathcal{B}\right]. \Pr[\mathcal{A}\land \mathcal{B}]
\]
and Equations~\eqref{event-A},~\eqref{event-B},~\eqref{z-bound}, we can conclude that
\[
\Pr\left[Z> \frac{n^{2\epsilon}}{256}\right]=1-o(1).
\]
Therefore, after one round there exist $\Omega(n^{2\epsilon})$ red nodes w.h.p.
\end{proof}

\paragraph{Constant Probability.} Theorem~\ref{er-thm} asserts that there is a constant probability that no node becomes red during the process, except one node which is initially red. We claim that if we replace ``constant probability'' with ``w.h.p.'', the statement is no longer true. Let $p=1/n^{(1/2)+\epsilon}$ for some $\epsilon>0$ and assume that $v$ is the only node which is red initially. We prove that with a constant probability at least one node becomes red after one round.

Consider the event $\mathcal{A}$ as defined in the proof of Theorem~\ref{er-tight}. Let $u$ be a neighbor of $v$. If $\mathcal{A}$ holds, then $|N(v)\cup N(u)|\le 4(n-1)p$. Thus, we have
\[
\Pr[\mathcal{C}_1(v)=r|\mathcal{A}]\ge \frac{2}{2\cdot 4(n-1)p}=\frac{1}{4(n-1)p}.
\]
Let $\mathcal{Q}$ be the event that $v$ does not make any node red in the first round. Using the fact that if $\mathcal{A}$ holds, then $d(v)\ge \frac{1}{2}(n-1)p$, we get
\begin{align*}
    \Pr[\mathcal{Q}|\mathcal{A}]\le & \left(1-\frac{1}{4(n-1)p}\right)^{\frac{1}{2}(n-1)p} \\ \le & \exp\left(-\frac{(n-1)p}{2\cdot 4(n-1)p}\right) \\ = & \exp\left(-\frac{1}{8}\right)
\end{align*}
where we used the estimate $1-x\le \exp(-x)$.

Recall that according to Equation~\eqref{event-A}, we have that $\Pr[\mathcal{A}]=1-o(1)$. Combining this with the above inequality, we get
\begin{align*}
    \Pr[\mathcal{Q}] & =\Pr[\mathcal{Q}|\mathcal{A}]\cdot \Pr[\mathcal{A}]+ \Pr[\mathcal{Q}|\bar{\mathcal{A}}]\cdot \Pr[\bar{\mathcal{A}}] \\ \le & \exp(-1/8) (1-o(1))+ 1\cdot o(1)\le C
\end{align*}
for some constant $0<C<1$. Thus, we have $\Pr[\bar{\mathcal{Q}}]\ge 1-C>0$, i.e., there is a constant probability that at least one node becomes red in the first round.

\section{Proof of Theorem~\ref{flower-thm}}
\label{appendix-flower}
A path of super nodes is a sequence of super nodes which form a path in the cycle obtained from collapsing each super node into a node. A path is \textit{uncolored} if all its super nodes are uncolored. In a \textit{reddish path}, there are no two adjacent uncolored super nodes and the endpoints are colored. We note that for any coloring of the ($n$,$r$)-flower graph, there is a set of maximal uncolored and reddish paths which partition the nodes in the graph. In the rest of this proof, any time we refer to a path, it is a path in this unique set of paths, where the coloring is clear from the context. Now, let us make the following observation, which comes in handy later.

\begin{obs}\label{obs-reddish-path}
Initially there are at most $s(n)$ reddish paths, and during the process the number of reddish paths stays the same or decreases. (This is because the reddish paths can join each other, but cannot split into smaller paths.)
\end{obs}

Let a \textit{phase} be a sequence of $k$ rounds. To analyze the process, we break it into phases rather than rounds. Let $\mathcal{C}$ be the coloring at the beginning of phase $i\in \mathbb{N}$. Consider all the endpoints of the uncolored paths and let $U$ be the boundary nodes in these endpoints. Let $\mathcal{E}_i$ be the event that no node in $U$ becomes red during the whole phase $i$.

Note that if the event $\mathcal{E}_i$ occurs, then all boundary nodes on the endpoints of the reddish paths become orange. Therefore, all nodes which are not on any reddish path will remain uncolored forever since they have no red neighbor. Our goal is to prove that w.h.p. this happens while still most of the super nodes are not on any reddish path, which implies that the rumor does not spread.

Let us calculate the probability $\Pr[\bar{\mathcal{E}}_i|\bar{\mathcal{E}}_{i-1}\land \cdots \land \bar{\mathcal{E}}_1]$.
Consider a node $u$ which is in $U$. By definition, it has at most one red neighbor $u'$ and $|\hat{N}(u)\cap \hat{N}(u')|=2$, $|N(u)\cup N(u')|= 2r+2$. Thus, the probability that $u'$ does not make $u$ red during the whole phase is at least $(1-\frac{2}{(2r+2)\cdot 2})^{k}\ge (\frac{3}{4})^{k}$, where we used $r\ge 1$. Since $k$ is a constant, this probability can be bounded by a constant $0<C<1$. Using Observation~\ref{obs-reddish-path}, we have that $|U|\le 2s(n)$. Hence, we can conclude that


\begin{align*}
    \Pr[\bar{\mathcal{E}}_i|\bar{\mathcal{E}}_{i-1}\land \cdots \land \bar{\mathcal{E}}_1]\le 1-C^{2s(n)}.
\end{align*}

Let us define $t^*:= (1/C)^{2s(n)} \log(n)$, then using the estimate $1-x\le \exp(-x)$, we get

\begin{align*}
\Pr[\land^{t^*}_{i=1}\bar{\mathcal{E}}_{i}] = &\Pr[\bar{\mathcal{E}}_{t^*} |\land^{t^*-1}_{i=1}\bar{\mathcal{E}_i}]\cdots \Pr[\bar{\mathcal{E}}_1]\\
& \leq \left(1 - C^{2s(n)} \right)^{t^*}\\ &\leq \exp\left(-t^*\cdot C^{2s(n)}\right) \\ & = \frac{1}{n}.
\end{align*}

Therefore, w.h.p. after at most $t^*$ phases (i.e., $kt^*$ rounds), we reach a coloring where all nodes which are not on any reddish path remain uncolored forever. To finish the proof, it only remains to show that the number of nodes which are on reddish paths after $kt^*$ rounds is small.

We note that we initially have at most $2s(n)$ super nodes on the reddish paths (since from every two adjacent super nodes on a reddish path, at least one is not uncolored). Each reddish path can potentially grow from both sides in each round. Thus, after $t^*$ phases (i.e., $kt^*$ rounds), the number of nodes on the reddish paths is at most $r \cdot \left(2s(n) \cdot kt^* + 2s(n) \right)$. Since $s(n)=o(\log n)$, we have $(1/C)^{2s(n)}=\mathcal{O}(n^{\epsilon}/\log^3 n)$. Now, using $r\le n^{1-\epsilon}$ and $k$ being a constant, we conclude that the number of such nodes is upper-bounded by 
$\mathcal{O}(n/\log n)$. Thus, the number of nodes on the reddish paths is sublinear (i.e., the rumor does not spread).

\section{Proof of Lemma~\ref{spreadfastEstar}}
\label{appendix-lemma-clique}

We split the proof into two parts. First, we prove that after one round at least $\log^2 n/260$ nodes are red in $\mathcal{K}$ and then, we show that one round after that all nodes are red or orange. w.p. $1-o(1/n)$.

Consider two nodes $v,w\in \mathcal{K}$ such that $v$ is red and $w$ is uncolored. Let $q^*$ be the probability that $v$ makes $w$ red in the next round. We have $|\hat{N}(v)\cap \hat{N}(w)|\ge \kappa$ and $|N(v)\cup N(w)|\le 4\kappa$. Thus, using $k=5$, we have
\begin{equation}
    \label{q*}
    q^*\ge \frac{|\hat{N}(v)\cap \hat{N}(w)|}{2^k|N(v)\cup N(w)|}\ge \frac{\kappa}{2^5\cdot 4\kappa}=\frac{1}{128}.
\end{equation}

\paragraph{Part I.} Let $v_1$ be the node which is colored red in round $t$ and let us label the other nodes in $\mathcal{K}$ from $v_2$ to $v_{\kappa}$. Define the Bernoulli random variable $x_i$ for $2 \leq i \leq \kappa$ to be 1 if and only if $v_i$ is colored red by $v_1$ in round $t+1$. Let $X:=\sum_{i=2}^{\kappa}x_i$ be the sum of independent random variables $x_i$'s. We have $\mathbb{E}[X]\ge (\kappa-1)q^*\ge (\log^2 n -1)/128\ge \log^2 n/130$ using Equation~\eqref{q*}, $\kappa\ge \log^2 n$, and $n$ being large. Now, applying Chernoff bound (Lemma~\ref{Chernoff} in Appendix~\ref{appendix-ineq}), we get $\Pr\left[X\le \frac{\log^2 n}{260}\right]\le \exp\left(-\frac{\log^2 n}{1040}\right)=o\left(\frac{1}{n}\right).$

\paragraph{Part II.} Let $U$ and $R$ be the set of uncolored and red nodes in $\mathcal{K}$ in round $t+1$ and $|R|\ge (\log^2 n)/260$. Consider an arbitrary node $u$ in $U$. Since each node in $R$ makes $u$ red independently w.p. $q^*$, the probability that $u$ is not colored red in the next round is at most $(1-q^*)^{|R|}\le (127/128)^{(\log^2 n)/260}$, where we used Equation~\eqref{q*} and $|R|\ge (\log^2 n)/260$. With a union bound, the probability that there exists a node in $U$ which does not become red in the next round is at most $ |U|\cdot (127/128)^{(\log^2 n)/260} \le n/C^{\log^2 n}$ for some constant $C>1$. Thus, this probability is at most $o(1/n)$.Combining Parts I and II implies our desired statement.

\section{Proof of Lemma~\ref{constantboundary}}
\label{appendix-lemma-constantboundary}

To prove Lemma~\ref{constantboundary}, we first need to prove Lemma~\ref{nodesonboundary}, which builds on two well-known Lemmas~\ref{expander mixing lemma} and~\ref{min edges}.

For two node set $A,B\subset V$, we define $e(A,B):=|\{(v,v')\in A \times B:\{v,v'\}\in E\}|$, where $A\times B$ is the Cartesian product of $A$ and $B$.

\begin{lemma}[\cite{friedman2003proof_ALONEIGVALUE}]
\label{expander mixing lemma}
For any two node sets $A,B$ in an $N$-node $D$-regular graph, $\left|e(A,B) - \frac{|A||B|D}{N}\right| \leq \lambda \sqrt{|A||B|}$.
\end{lemma}

\begin{lemma}[\cite{friedman2003proof_ALONEIGVALUE}]
\label{min edges}
In an $N$-node $D$-regular graph $G$, for any two disjoint node sets $A,B$

\begin{equation*}
    e(A,B) \geq \frac{(D-\lambda)|A||B|}{N}
\end{equation*}
\end{lemma}

\begin{lemma}
\label{nodesonboundary}
Consider an $N$-node $D$-regular graph $G$ with $\lambda < D$. For every node set $A \subset V$, $|\partial(A)|$ is at least

\begin{equation*}
\min\left( \frac{(N-|A|)(D-\lambda)}{2D}, \frac{|A|}{4}\left(1-\frac{|A|}{N}\right)^2\left(\frac{D}{\lambda}-1\right)^2\right)
\end{equation*}
\end{lemma}

\begin{proof}
Let $\bar{A} := V\setminus A$ and $\partial :=|\partial(A)|$. Then, according to Lemma \ref{min edges},
\begin{equation}\label{eq1-partial}
\frac{|A|(N-|A|)}{N}(D-\lambda) \leq e(A,\bar{A}).   
\end{equation}
Furthermore, according to Lemma \ref{expander mixing lemma} we have
\begin{equation}\label{eq2-partial}
    e(A, \partial(A)) \leq \frac{|A| \partial D}{N} + \lambda \sqrt{|A| \partial}.
\end{equation}

Combining Equations~\eqref{eq1-partial},~\eqref{eq2-partial} and using the fact that $e(A,\bar{A}) = e(A, \partial(A))$, we conclude that
\begin{align*}
    \frac{|A|(N-|A|)}{N}(D-\lambda) \leq \frac{|A| \partial D}{N} + \lambda\sqrt{|A|\partial} \iff \\ \left(N-|A|\right)(D-\lambda)\le \partial D+ N\lambda \sqrt{\frac{\partial}{|A|}}.
\end{align*}

We note that if $\partial D \geq N\lambda \sqrt{\frac{\partial}{|A|}}$, then
\[
(N-|A|)(D-\lambda) \leq 2 \partial D \iff \frac{(N-|A|)(D-\lambda)}{2D} \leq \partial.
\]

If $\partial D < N\lambda \sqrt{\frac{\partial}{|A|}}$, then
\begin{align*}
(N-|A|)(D-\lambda) \leq 2N\lambda \sqrt{\frac{\partial}{|A|}} \iff \\ \frac{|A|}{4}\left(1-\frac{|A|}{N}\right)^2\left(\frac{D}{\lambda}-1\right)^2 \leq \partial.
\end{align*}
\end{proof}

\textit{Proof of Lemma~\ref{constantboundary}.} Lemma~\ref{nodesonboundary} implies the two following inequalities. Firstly, we note that
\begin{align*}
    \frac{(N-|A|)(D-\lambda)}{2D} \geq & \frac{(\frac{9}{10}N)(D-C\sqrt{D})}{2D} \\ \geq & \frac{(\frac{9}{10}N)(\frac{9}{10}D)}{2D} \geq \frac{2}{5}N
\end{align*}
where we used $|A|\le N/10$, $\lambda\le C \sqrt{D}$, and $D=\omega(1)$. Secondly, we have
\begin{align*}
    \frac{|A|}{4}\left(1 - \frac{|A|}{N}\right)^2\left(\frac{D}{\lambda}-1\right)^2 \geq & \frac{|A|}{4}\left(\frac{9}{10}\right)^2\left(\frac{\sqrt{D}}{C} -1\right)^2
   \\  \geq & C'|A|D
\end{align*}

for some constant $C' > 0$, where we again used $|A|\le N/10$, $\lambda\le C \sqrt{D}$, and $D=\omega(1)$. \qed

\section{Proof Sketch of Theorem~\ref{me-thm}}
\label{appendix-me-thm}
Here, we provide a proof sketch for Theorem~\ref{me-thm}. Let the super node containing the initially red node be $x$. According to Lemma~\ref{spreadfastEstar}, after two rounds, all nodes in $x$ are red w.h.p.

We want to apply Lemma~\ref{me-lemma} repeatedly until we reach at least $C_1N/D$ strong red super nodes. As the base case, we can apply the lemma for $t=2$ since the super node $x$ is strong red. Assume that we have applied the lemma for some $t_0$, and now want to apply it for $t_1=t_0+3$ to show that $s_{t_2}=\Omega(s_{t_1}d)$ for $t_2=t_0+6$. To apply the lemma, the condition $w_{t_1}=\mathcal{O}(s_{t_1}/d)$ needs to be satisfied. We know that $w_{t_1}\le s_{t_0}+w_{t_0}\le s_{t_0}+\mathcal{O}(s_{t_0}/d)\le 2s_{t_0}$. Furthermore, $s_{t_1}=\Omega(s_{t_0}d)$ implies that $s_{t_0}=\mathcal{O}(s_{t_1}/d)$. Combining the last two statements gives $w_{t_1}\le \mathcal{O}(s_{t_1}/d)$.

Therefore, after $3t^*$ rounds for some $t^*=\mathcal{O}(\log_d n)$, the process reaches at least $C_1N/D$ strong red super nodes with the error probability smaller than
\begin{multline*}
\sum_{i=1}^{t^*} \exp(-\Omega(d^i))+\sum_{i=1}^{t^*}o\left(\frac{1}{\log n}\right)  \\ \le  \sum_{i=1}^{t^*}\frac{1}{\Omega(d^i)}+\mathcal{O}(\log_d n)\cdot o\left(\frac{1}{\log n}\right) \\ \le  \mathcal{O}\left(\frac{1}{d}\right)+o(1) = o(1)
\end{multline*}
where we used that the first sum is a geometric series and $d=\omega(1)$. Hence, we can conclude that after $\mathcal{O}(\log_d n)$ rounds, there will be at least $C_1N/D$ strong red super nodes w.h.p.

There is one detail which was left out in the above argument. In addition to $\Omega(s_td)$ newly generated strong red nodes (according to Lemma~\ref{me-lemma}), some super nodes might get red but not fully red during a three-round phase (i.e., only a strict subset of their nodes become red). Intuitively speaking, such red nodes will contribute to the spread of the rumor, which is what we are after. However, to be completely accurate, we need to take such super nodes into account in our calculations, but we did not for the sake of simplicity.

Finally, one can prove that from a coloring with $\Omega(N/D)$ strong red super nodes, the process reaches at least $N/10$ strong red super rounds in a few more rounds w.h.p. This can be proven using Lemma~\ref{min edges} and an argument similar to the one in the proof of Lemma~\ref{spreadfastEstar} or the tightness of Theorem~\ref{er-thm}, which is left out to avoid redundancy.

\paragraph{Tightness.} We prove that the condition $d=\omega(1)$ is necessary by proving that if $d$ is a constant, then there is a constant probability that the rumor does not spread. Assume that all nodes in a super node $x$ are red (and the rest is uncolored). Consider a node $v$ in $\partial(x)$. The probability that $v$ is made red by its neighbor, say $u$, in $x$ (note that according to Observation~\ref{obvs1}, it has exactly one neighbor in $x$) is at most $d/\log^2$ since $|\hat{N}(v)\cap \hat{N}(u)|\le d$ and $|N(v)\cup N(u)|\ge d(v)\ge \log^2 n$. The probability that $v$ does not become red during the next $k$ rounds is at least $(1-d/\log^2 n)^k$. The probability that none of the nodes in $\partial(x)$ becomes red during the next $k$ round is at least $(1-d/\log^2 n)^{kD}\ge 4^{-(kDd)/\log^2 n}=4^{-kd^2}$ which is a non-zero constant probability, when $d$ is a constant. (We used the estimate $1-x\ge 4^{-x}$ for $0\le x\le 1/2$ and $D=d\log^2 n$.) Hence, with a constant probability we reach the configuration where only the nodes in $x$ are orange and the rest of nodes are uncolored (i.e., the rumor does not spread).

\section{Additional Experimental Results}
\label{appendix-plots}
We presented the outcome of our experiments on FB and TW SNs In Figure~\ref{ExperimentsIntext}. In Figure~\ref{appendixExperiments}, we provide similar results for G+, T-GE, and T-FR. (The countermeasure CM2 is only run on T-GE and T-FR due to the computational costs of running this countermeasure on G+.) Additionally, we provide the maximum and average standard deviation of the experiments in Table~\ref{tablestd}. 
\FloatBarrier
\begin{figure}[!h]
\def\tabularxcolumn#1{m{#1}}

\begin{tabularx}{\linewidth}{X}
\begin{tabular}{cc}
\subfloat[]{\includegraphics[width=4.5cm]{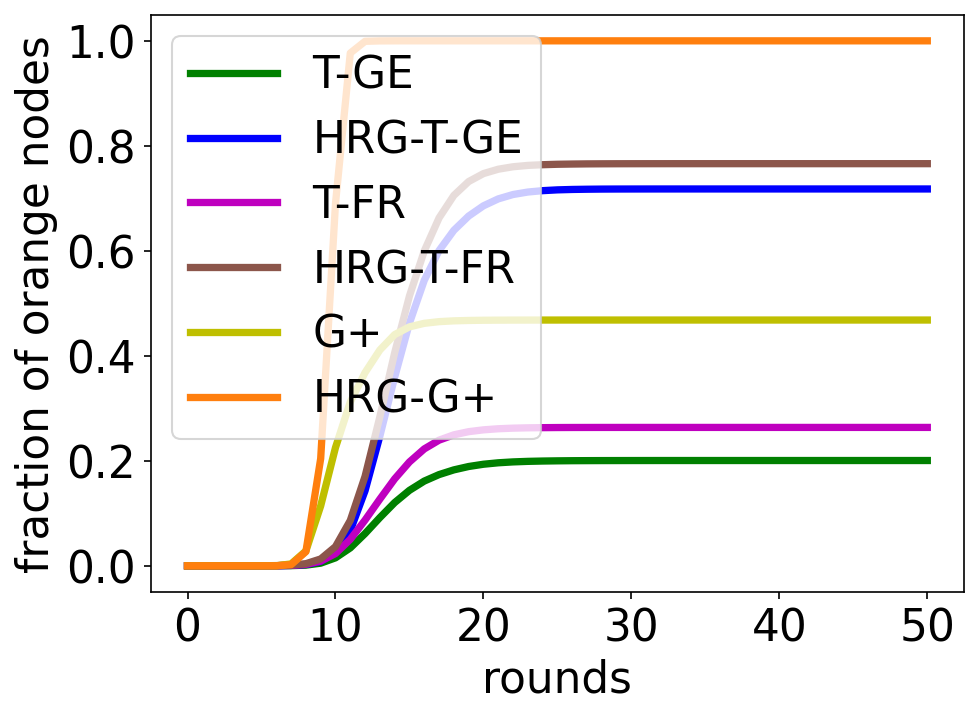}} 
\subfloat[]{\includegraphics[width=4.5cm]{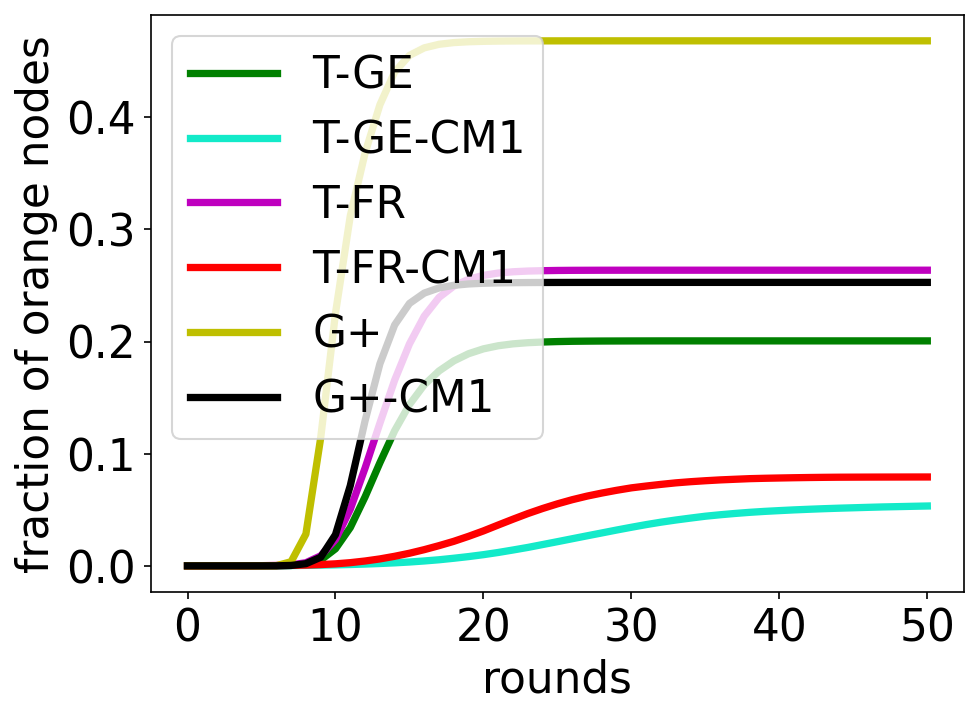}}\\
\subfloat[]{\includegraphics[width=4.5cm]{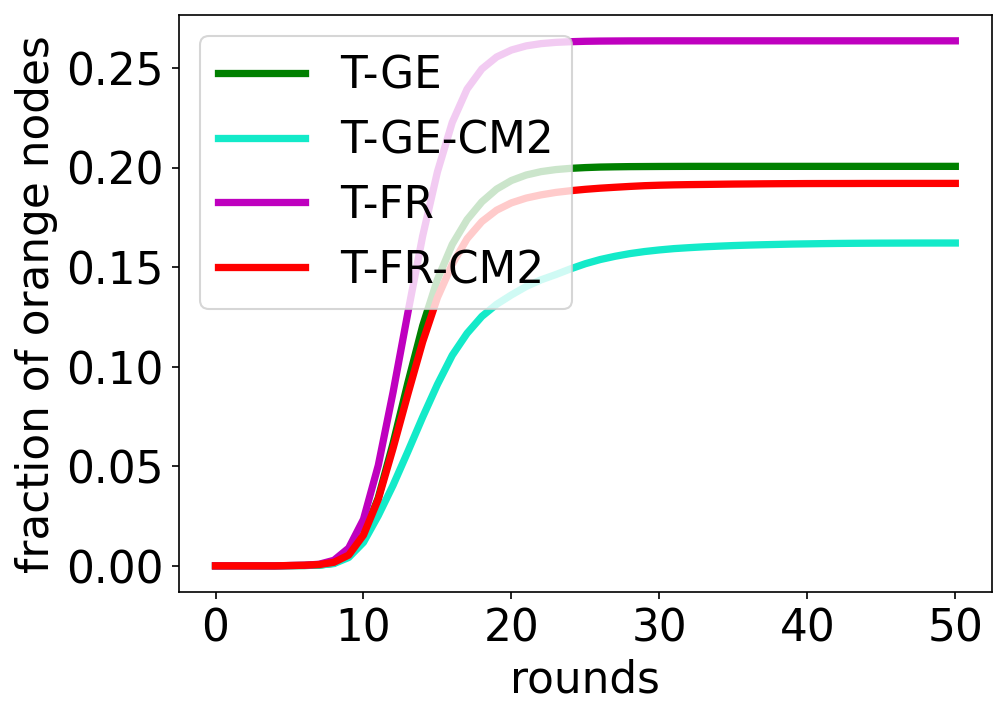}}
   \subfloat[]{\includegraphics[width=4.5cm]{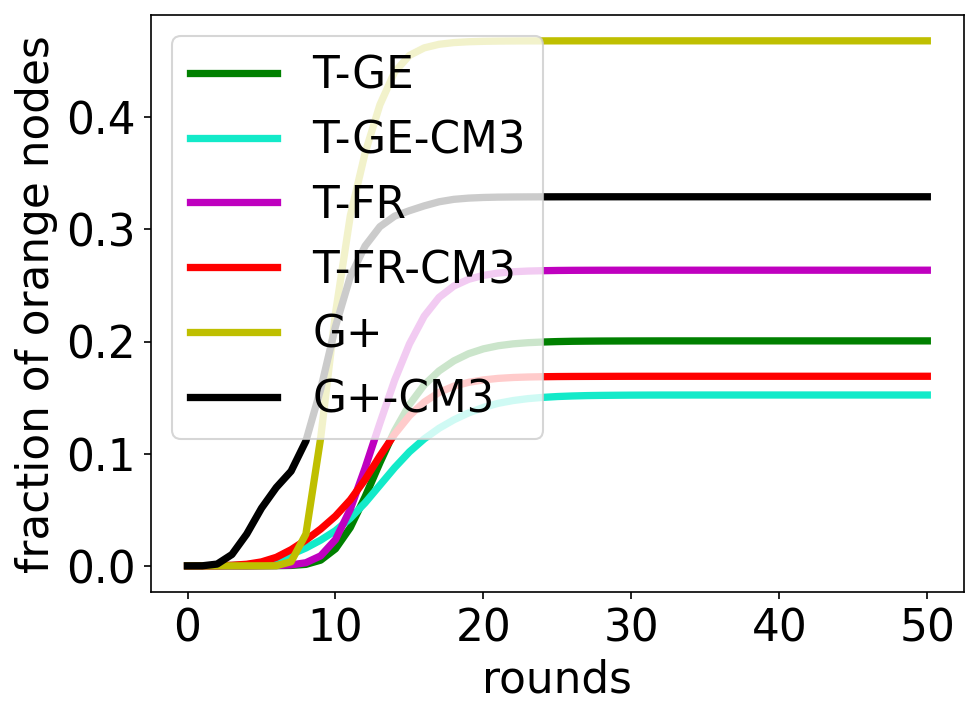}} \\
    \subfloat[]{\includegraphics[width=4.5cm]{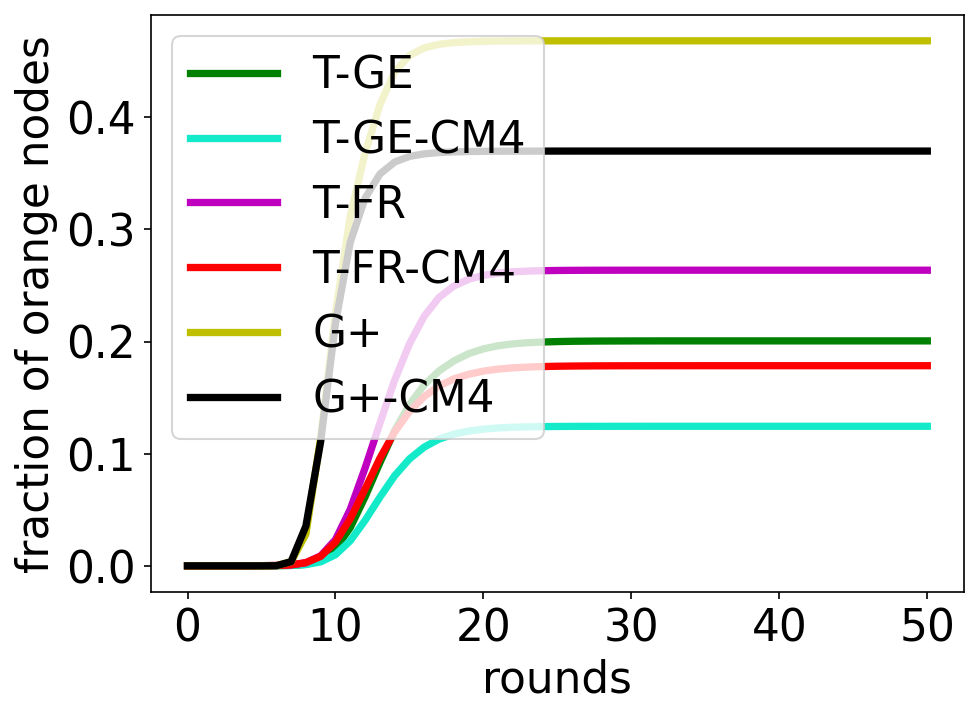}}
    \subfloat[]{\includegraphics[width=4.5cm]{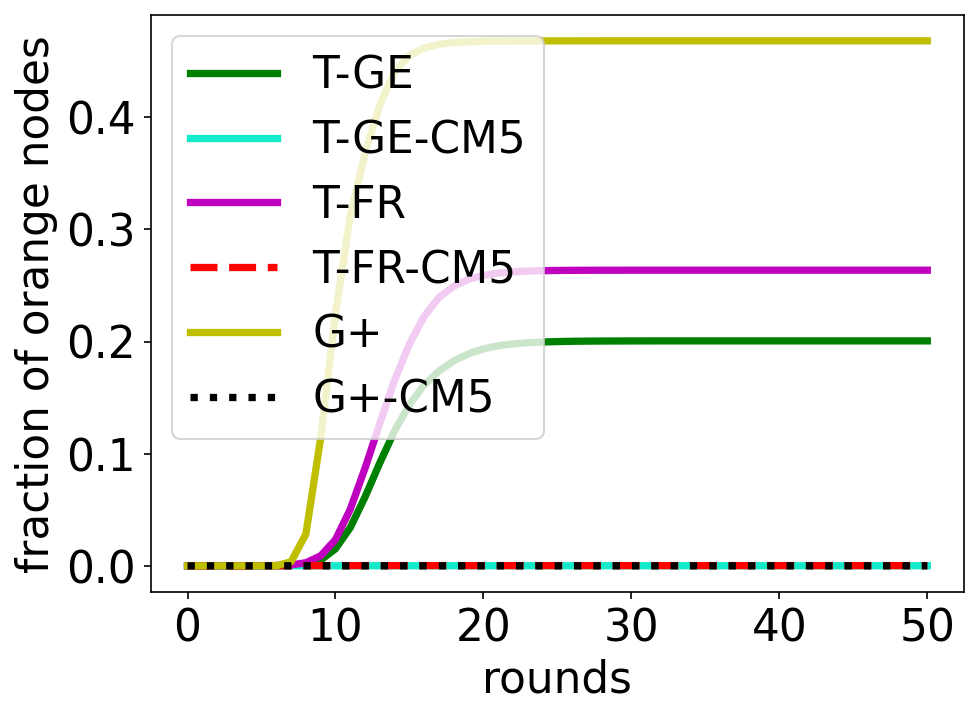}}\\
    \subfloat[]{\includegraphics[width=4.5cm]{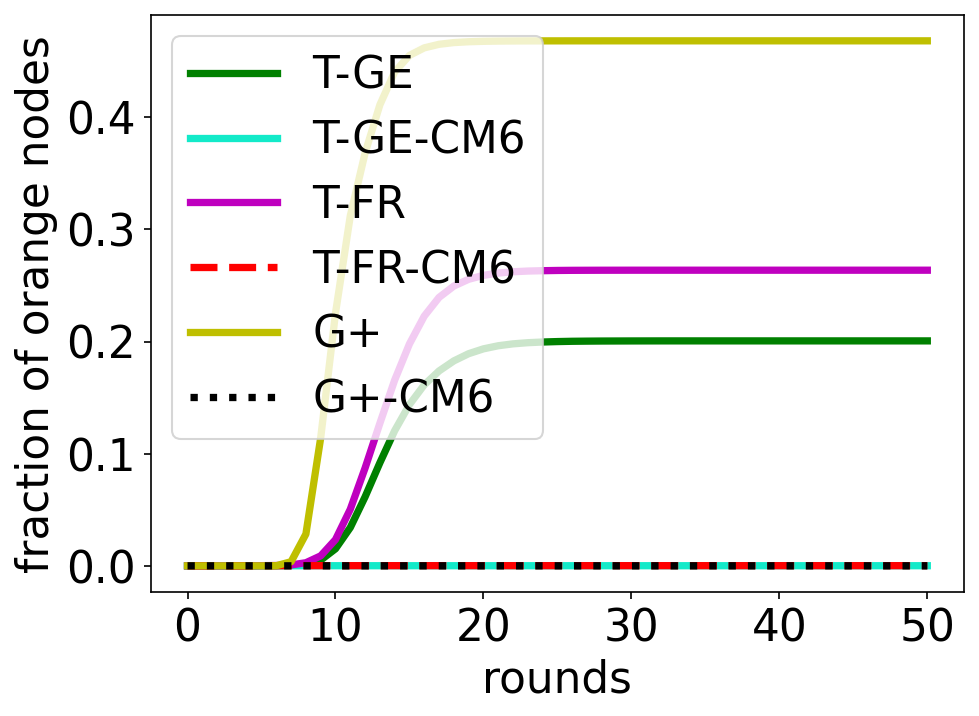}}
\end{tabular}
\end{tabularx}

\caption{Fraction of orange nodes, from one randomly chosen red node, on G+, T-GE, T-FR and (a) HRG with comparable parameters and (b-g) after the implementation of countermeasures CM1 to CM6.}\label{appendixExperiments}
\end{figure}

\section{Impact of Delay in CM4}
\label{appendix-delay}

In Figure~\ref{ExperimentsIntext}-(f), the outcome of our experiments are depicted for the countermeasure CM4 when the delay parameter $\tau$ is equal to 4. In Figure~\ref{CM4delays}, the outcome of the experiments for different values of $\tau$ are visualized. We observe, as one might expect, the final fraction of orange nodes increases in $\tau$ since the later we start spreading the anti-rumor, the more the rumor spreads. However, even for $\tau=1$ (which implies that a very robust rumor detection strategy is in place that can spot the rumor immediately), the rumor spreads to a large body of the network.

\FloatBarrier
\begin{figure}[ht]
\def\tabularxcolumn#1{m{#1}}
\hspace*{-1cm}
\begin{tabularx}{\linewidth}{cc}
\begin{tabular}{cc}
\subfloat[FB]{\includegraphics[width=4.5cm]{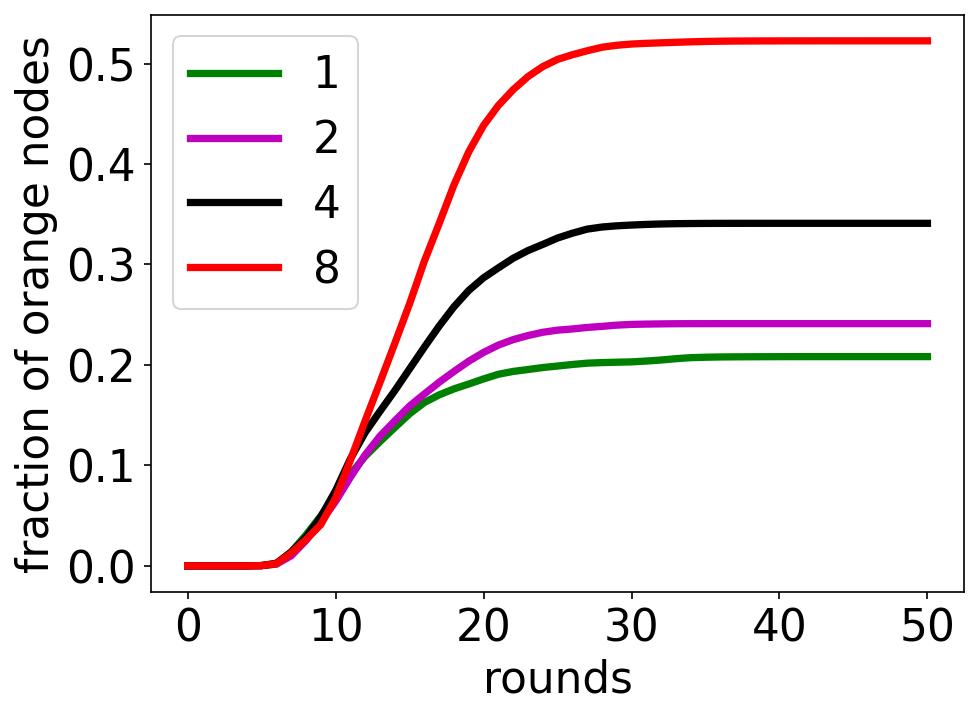}} 
\subfloat[TW]{\includegraphics[width=4.5cm]{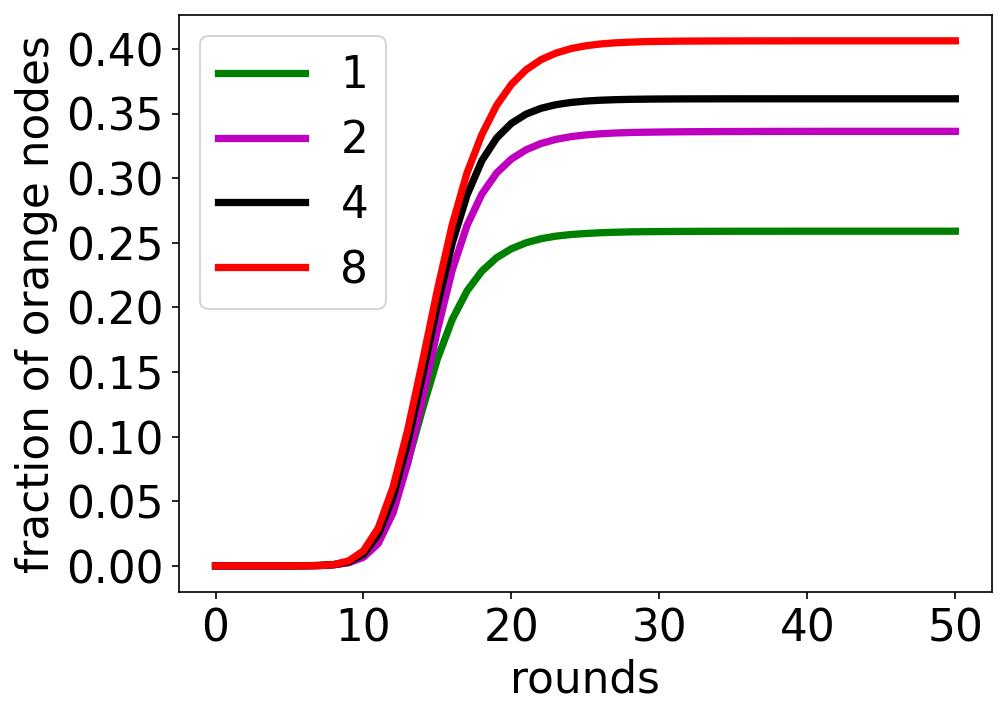}}
\end{tabular}
\end{tabularx}

\caption{Countermeasure CM4 with different values of delay parameter $\tau$ applied to (a) FB, (b) TW.}\label{CM4delays}
\end{figure}
\FloatBarrier
\begin{table}[t]
\begin{tabular}{ |p{1.4cm}||p{3.5cm}|p{3.6cm}|}
 \hline
  & max std & av std \\
 \hline
 Flower   &  0.0003570714214271427 & 0.00019604979682331548  \\
  \hline
 ME-low   & 0.17915726335837603 & 0.09321301350557022\\
  \hline
 ER-high   & 0.19308936724674927 & 0.0859734160271039\\
  \hline
 ER-low   & 0.0019283692106285047 & 0.0011524607771197107\\
  \hline
 FB   & 0.3337923725486565 & 0.2461821571260035\\
 \hline
 HRG-FB  & 0.2829449495964113 & 0.18586474914826873\\
 \hline
 TW & 0.2877416333455981 & 0.2264068652466434\\
 \hline
 HRG-TW & 0.32001708571888393 & 0.24086638967425303 \\
 \hline
 FB-CM1 & 0.23726334289357218 & 0.1621547695015237 \\
 \hline 
 TW-CM1 & 0.2287326883235253 & 0.16046318390388867 \\
 \hline 
 ME & 0.1987321404868875 & 0.03284228340220624 \\
 \hline
 ME-CM1 & 0.20278712246755684 & 0.1183214121824717 \\
 \hline 
 FB-CM2 & 0.14736503430973327 & 0.0952178187923975 \\
 \hline 
 ME-CM2 & 0.12251631531622152 & 0.09461957732147125 \\
 \hline 
 FB-CM3 & 0.27388679494130697 & 0.2006883240028243 \\
 \hline 
 TW-CM3 & 0.2828668855500441 & 0.21307654532259354 \\
 \hline 
 ME-CM3 & 0.3198130728449042 & 0.2258139637381856 \\
 \hline 
 FB-CM4 & 0.27594897712149935 & 0.21358225073276074 \\
 \hline 
 TW-CM4 & 0.3002198524587406 &0.2351432839236312 \\
 \hline 
 ME-CM4 & 0.18863712063443397 & 0.04379009270170083 \\
 \hline 
 FB-CM5 & 0.000340338505123834 & 0.00032744174880541094 \\
 \hline 
 TW-CM5 & 1.9105792826203032e-05
 & 1.8823414951313245e-05\\
 \hline 
 ME-CM5 & 0.0004172154719087009 & 0.0003354160313820336 \\
 \hline 
 FB-CM6 & 0.0 & 0.0 \\
 \hline 
 TW-CM6 & 0.0 & 0.0 \\
 \hline 
 ME-CM6 & 0.0 & 0.0 \\
 \hline 
 
\end{tabular}
\caption{The standard deviations in our experiments.}
\label{tablestd}
\end{table}

\section{Proof of Theorem~\ref{2timesME}}
\label{appendix-2timesME}
We prove that w.h.p. there is no node (outside $x$ and $y$) which has two neighbors in $x$ and $y$. This implies that no node outside $x$ and $y$ will ever become red. Thus, at most $2\log^2 n$ nodes (the nodes in $x$ and $y$) become red (and then orange) during the process, i.e., the rumor does not spread.

Note that by construction of moderate expander graphs, a node outside $x$ and $y$ cannot have more than one edge to $x$ (or to $y$). Thus, for a node to become red, it must have a neighbor in $x$ and a neighbor in $y$. An arbitrary node $v$ is adjacent to $d$ super nodes. The probability that the chosen $x$ and $y$ are in its neighborhood is
\[
\frac{{d \choose 2}}{{N \choose 2}}\le \frac{2d^2}{N^2}\le\frac{2n^{1-2\epsilon}\log^4 n}{n^2}=o\left(\frac{1}{n}\right).
\]
A union bound over all $n$ nodes implies that w.h.p. there is no node outside $x$ and $y$ which has two neighbors in the union of $x$ and $y$.

\section{Entries of Table~\ref{table}}
\label{appendix-table}
Let us start with C1. According to Figure~\ref{ExperimentsIntext}, it is easy to observe that CM1, CM2, CM3, and CM4 are not very effective. While they reduce the extent that the rumor spreads, it still spreads to a large fraction of the network. On the other hand, countermeasures CM5 and CM6 stop the spread of the rumor very effectively.

CM6 satisfies the criterion C2 since it simply requires the agents to spread a piece of information only when they have heard it twice. We have set ``jein'' for other countermeasures since they are not extremely hard to execute, but definitely require smart and careful implementation of some strategies. Most of them require a functional rumor detection procedure to be in place. While several algorithms, using techniques from NLP, have been proposed~\cite{dharod2021trumer}, the rumor detection is an infamously difficult and complex task. Furthermore, the algorithms to block nodes and edges in CM1 and CM2 need the full knowledge of the network.

CM1 and CM2 clearly do not support C3 since they require blocking agents or their connections. CM3 attempts to provide the users with extra relevant information, but would not intervene with their freed of expression. CM4 and CM5 do not intervene with the freedom of expression either but rather use that to spread the truth. CM6 requires educating the agents to express their opinions more wisely and patiently, but does not forbid them from doing so.

CM1 and CM2 clearly do not satisfy C4 since they change the network structure radically by removing a substantial number of edges/nodes. CM4 and CM5 manipulate the process to some extent, but they are not as intrusive as the first two countermeasures. One can argue CM3 and CM6 are even less intrusive.
\end{document}